\documentclass[letterpaper,12pt]{amsart}

\usepackage{latexsym, amsmath, amssymb,amsthm, natbib,verbatim, xy}
\xyoption{all}

\usepackage{tgpagella}
\usepackage{mathpazo}

\usepackage[left=.95in,top=1in,right=.95in,bottom=1in]{geometry}
\usepackage{setspace,color}
\usepackage{graphicx}
\usepackage{enumerate}
\usepackage[multiple]{footmisc}
\usepackage[leqno]{amsmath}
\usepackage{hyperref}
\usepackage{mathtools}

\usepackage{pgf,tikz}
\usepackage{graphicx}
\usetikzlibrary{arrows}
\definecolor{ffqqqq}{rgb}{1.,0.,0.}
\definecolor{ududff}{rgb}{0.30196078431372547,0.30196078431372547,1.}

\hypersetup{colorlinks,
            linkcolor=blue,
            anchorcolor=blue,
            citecolor=blue,urlcolor={blue}}

\usepackage{etoolbox}
\patchcmd{\section}{\scshape}{\bfseries}{}{}
\makeatletter
\renewcommand{\@secnumfont}{\bfseries}
\makeatother

\renewenvironment{quote}{%
   \list{}{%
     \leftmargin0.75cm   
     \rightmargin0.5cm
   }
   \item\relax
}
{\endlist}

\usepackage{tikz}

\newtheorem{theorem}{Theorem}
\newtheorem{corollary}{Corollary}
\newtheorem{lemma}{Lemma}
\newtheorem{proposition}{Proposition}

\theoremstyle{definition}

\newtheorem{example}{Example}

\renewenvironment{quote}{%
   \list{}{%
     \leftmargin0.75cm   
     \rightmargin0.5cm
   }
   \item\relax
}
{\endlist}

\begin{document}

\title{Queueing games with an endogenous number of machines}

\author[Atay and Trudeau]{Ata Atay \and Christian Trudeau}\thanks{
Ata Atay is a Serra H\'{u}nter Fellow (Professor Lector Serra H\'{u}nter). Ata Atay gratefully acknowledges financial support by the University of Barcelona through grant AS017672. Christian Trudeau gratefully acknowledges financial support by the Social Sciences and
Humanities Research Council of Canada [grant number 435-2019-0141].\\
Atay: Department of Mathematical Economics, Finance and Actuarial Sciences, and Barcelona Economic Analysis Team (BEAT), University of Barcelona, Spain. E-mail: \href{mailto:aatay@ub.edu}{aatay@ub.edu}. \\
Trudeau: Department of Economics, University of Windsor, Windsor, ON, Canada. E-mail: \href{mailto:trudeauc@uwindsor.ca}{trudeauc@uwindsor.ca}.}

\date{\today}
\maketitle
\begin{abstract}
This paper studies queueing problems with an endogenous number of machines with and without an initial queue, the novelty being that coalitions not only choose how to queue, but also on how many machines. For a given problem, agents can (de)activate as many machines as they want, at a cost. After minimizing the total cost (processing costs and machine costs), we use a game theoretical approach to share to proceeds of this cooperation, and study the existence of stable allocations. First, we study queueing problems  with an endogenous number of machines, and examine how to share the total cost. We provide an upper bound and a lower bound on the cost of a machine to guarantee the non-emptiness of the core (the set of stable allocations). Next, we study requeueing problems with an endogenous number of machines, where there is an existing queue. We examine how to share the cost savings compared to the initial situation, when optimally requeueing/changing the number of machines. Although, in general, stable allocation may not exist, we guarantee the existence of stable allocations when all machines are considered public goods, and we start with an initial schedule that might not have the optimal number of machines, but in which agents with large waiting costs are processed first. 
\end{abstract}
\newpage
\section{Introduction}
\label{sec:intro}
Consider a set of agents with jobs that have to be executed by a number of machines in such a way that the aim is to minimize the total cost based on some criterion. We observe such problems in many real-life applications such as manufacturing, health care, logistics, etc. In this paper, we consider queueing problems from two different perspectives; (i) queueing problems that consider the problem of optimally queueing the agents before they arrive, (ii) queueing problems that consider the problem of reorganizing (rescheduling) an existing queue optimally. In both problems, a set of agents wait for their jobs to be processed on machines. Each agent has a job that needs the same amount of processing time with a different unit waiting cost. We refer to \cite{chun16} for a comprehensive survey on queueing theory. 

This paper is the first one that allows for an endogenous number of machines. It thus includes the tradeoff that groups have between the cost of maintaining multiple machines and the savings of having their jobs processed faster on said machines. As an example, during the COVID pandemic, health authorities not only had to decide on the order of the queue for vaccines, but also on the speed of the vaccination operations. Similarly, research groups have to determine if they prefer to wait for access to highly-specialized equipment or to buy new equipment for faster access. The concept of an endogenous number of machines is particularly relevant when studying, as we do, the problem using cooperative game theory; the concept of core stability now implies that when a coalition threatens to leave the group, it would do so by paying for for the number of machines that minimizes its own cost.

\cite{m03} studied one machine queueing problem from a cooperative game theoretical perspective and showed that the rule assigning positions in the queue and compensations is the \cite{s53} value of the associated TU-game. \cite{chun06} introduced a pessimistic definition of the worth that can be generated by a subset of agents. It is proved that different definitions lead to very different rules. Towards a generalization to multiple machines, \cite{cjh08} consider queueing problems with two parallel machines. \cite{cetal89} are the first to study one-machine sequencing problems from a cooperative game theoretical point of view. That is, queueing problems with an initial queue where rescheduling is allowed to improve upon the initial situation. The rescheduling of jobs is allowed to reduce weighted completion time and the total savings by rescheduling can be shared by agents who own the jobs. \cite{hetal99} (see also \citealp{s06a}) consider multiple parallel sequencing situations where the number of machines is fixed. They guarantee the non-emptiness of the core for one and two machine situations, and moreover for two subclasses when there are at least three machines.

By contrast, in this paper we consider that the number of machines is endogenous. Each machine has a cost to activate. Hence, in both types of problems, a subset of agents can ``buy'' as many machines as they want to in exchange of its cost, and under some constraints in problems with an existing initial order, might be able to ``sell'' some of the existing machines.  Moreover, each agent incurs some waiting cost until her job is processed and she can leave the system. Then, we take a game theoretical approach to address the question on how to distribute among the players the proceeds of their cooperation, whenever they (re)schedule their jobs to be processed in an optimal way, minimizing total costs. Following the vast literature on different problems on rescheduling an initial queue (see for instance \citealp{cetal01}; \citealp{metal15}; \citealp{bt19}; \citealp{aetal21}), we examine conditions guaranteeing the existence of stable allocations. 

First, we examine queueing problems that consider the problem of optimally queueing the agents before they arrive. Traditionally, queueing problems have a fixed number of machines, and their cost is sunk and thus ignored. The resulting cost game is then superadditive, as congestion implies that the total waiting cost for two agents is larger than the sum of their waiting costs if they are alone. With an endogenous number of machines, the cost function is always subadditive, as two agents can always each buy a machine, generating costs equal to the sum of their individual stand-alone costs. This allows a traditional definition of the core. We provide a lower bound and an upper bound on the cost of a machine for the non-emptiness of the core of a queueing problem with an endogenous number of machines (Theorems \ref{thm:queue_game_upper_bound}, \ref{thm:queue_game_lower_bound}). In the second case, we provide a full description of the core (Theorem \ref{thm:queue_game_lower_bound}).

Next, we consider the problem of rescheduling an existing queue optimally with an endogenous number of machines which we call the requeueing problem. While the sequencing literature has crucially depended on the assumptions regarding the admissible rearrangements of the initial queue by a coalition on a given machine, our results depend mostly on the assumptions made on the adjustment of the queue when new machines are added. 


If we suppose that a coalition will be the sole users of a new machine it adds, we see that the core of a requeueing game with an endogenous number of machines may be empty (Example \ref{ex:hurt}), and that under any assumptions on how a coalition can reorganize its members on the existing machines. Nevertheless, we provide sufficient conditions to guarantee the non-emptiness of the core (Theorems \ref{thm:lower_bound_sequencing}, \ref{thm:upper_bound_sequencing}). We also obtain positive results if we suppose that machines are public goods. That is, when we add machines, the whole queue moves up, and not only the members of the coalition that paid for the extra machines. The distinction is akin to establishing VIP machines and general-use machine. We show that under the assumption of public machines, whenever the initial schedule efficiently orders agents from high to low waiting costs but might not have the optimal number of machines, the public requeueing games always have a non-empty core (Theorem \ref{thm:public_convex}). 

The paper is organized as follows. In Section \ref{sec:queue} we present queueing problems with an endogenous number of machines. In Section \ref{sec:queue_games} we introduce the associated TU-game for queueing problems with an endogenous number of machines. We derive upper and lower bounds on the cost of a machine to guarantee the existence of stable allocations as well as a full characterization of the set of stable allocations. In Section \ref{sec:requeue} we introduce two types of problems and their associated TU-games, the so-called private and public requeueing problems (and games). For private requeuing games, although stable allocations need not exist, we provide an upper-bound and a lower-bound to guarantee their existence. For public requeuing games, we show that stable allocations always exist if the initial schedule serves agents with larger waiting costs first. For both types of games we make different assumptions on what a coalition is allowed to do to reorganize the initial queue and discuss the implications. Finally, we draw conclusions in Section \ref{sec:conc}. We consign proofs of lemmata on optimal number of machines in Appendix \ref{sec:appendix}.

\section{Queueing problems with an endogenous number of machines}
\label{sec:queue}
We examine first the queueing problem. We have a set of agents $N=\{1,2,\ldots, n\}$. When no confusion arises we denote by $|N|=n$ the cardinality of the set of agents. Each agent has one job to be processed on a machine. The agents have access to an unlimited number of machines, but they must pay $b\in\mathbb{R}_+$ for each machine that they use. All jobs and all machines are identical, and each machine can process one job per period.  We assume that each machine starts processing at time 0. 

Every agent $i\in N$ has a waiting cost that is linear with respect to the time it spends in the system. The waiting cost function of an agent $i\in N$ is $w_{i}t$ where $w_{i}>0$ is the waiting (weight) cost per unit time of player $i$ and $t$ is the period at which the job has been processed. We refer to the vector of weights by $w:= (w_{i})_{i\in N}$.  Let $w^{S}_{k}$ be the waiting cost of the $k^{th}$ agent (according to the order $N$) in $S$ and $w^{-S}_{k}\equiv w^{N\setminus S}_{k}$ be the waiting cost of the $k^{th}$ agent outside the coalition $S$. 

 A \emph{queueing problem with an endogenous number of machines} can be described as $(N,w,b)$ where $N$ is the set of agents, $w$ is the vector of unit waiting costs and $b\in\mathbb{R}_+$ is the cost of a machine. We suppose that $w_1\geq w_2\geq...\geq w_n$.

In a queueing problem, we examine the problem before agents arrive to queue: we are looking for the optimal number of machines and the optimal queueing of agents on those machines, the objective being the minimization of the total cost, consisting of the agents' waiting costs and the machine costs.

The solution consists in choosing a number of machines $m\in\{1,...,n\}$ and a schedule $\sigma=(\varphi, s)$, where $\varphi:N\rightarrow \{1,...,m\}$ assigns agents to machines and $s: N\rightarrow \mathbb{N}\cup \{0\}$ assigns to each agent a starting time. 
A schedule $\sigma=(\varphi, s)$ is admissible if for all $i,j\in N$, $\varphi(i)=\varphi(j) \Rightarrow s(i) \neq s(j)$. In words, if two agents are assigned to the same machine, they must have different starting times. 
The set of all possible schedules with $m$ machines is denoted by $\Sigma(m)$. A scheduling plan is $(m,\sigma)$, with $\sigma\in \Sigma(m)$.

Let $N^k=\{i\in N: \varphi(i)=k \}$ be the set of agents assigned to machine $k$. A schedule is a \emph{semi-active schedule} if there is no job which could be started earlier without altering the processing schedule. This has two implications for $\sigma$. First, we must have that if $\varphi(i)=k$ and $s(i)=l>0$, there must be $j\in N^k$ such that $s(j)=l'$ for all $l'\in \{0,...,l-1\}$. Second, we must have $|N^k|-|N^{k'}|<2$ for all $k,k'\in \{1,...,m\}$. In words, the first condition imposes that a schedule on a given machine has no downtime, and processes a job at all periods until all agents assigned to that machine have their job processed. The second condition imposes a difference in the number of agents assigned to pairs of machines to be at most one; otherwise, we could move the last agent on the first machine to the last position on the second machine, reducing the processing time of that agent without affecting the processing time of other agents.

Since no preemption is allowed, the completion time of the job of agent $i$ according to $\sigma=(\varphi, s)$ is $s(i)+1$. Hence, the waiting cost of an agent $i\in N$ can be written as $c_{\sigma}(i)=w_{i}(s(i)+1)$.

We thus need to find $(m,\sigma)$ that optimizes the following objective function: 
\[
\min_{m\in\{1,...,n\}}\left(bm+\min_{\sigma\in\Sigma^m}\sum_{i\in N}c_{\sigma}(i)\right).
\]

It is well-established in the literature that, for the one-machine case (with equal processing times), the total cost is minimal if the players are arranged according to their waiting costs in a decreasing order (see \citealp{s56}; \citealp{cetal89}). With multiple machines, it remains optimal to not process jobs of agents with larger waiting costs after those of agents with smaller waiting costs, i.e. $w_i<w_j \Rightarrow s(i)\leq s(j)$.

Given this result, if we install $m$ machines, it is optimal to schedule the $m$ agents with the highest waiting costs (agents $\{1,\ldots,m\}$) at time 0, and it is irrelevant to which machine each agent is assigned to. The next $m$ agents are then scheduled in the next period, and so on. Thus, the queueing problem reduces to finding the number of machines that solves
\[
\min_{m\in\{1,\ldots,n\}}\left(bm+\sum\limits_{i\in N}\left( \left\lceil \frac{i}{m}\right\rceil \right) w_{i}\right). \footnote{For all $x\in\mathbb{R}$, $\lceil x\rceil := \min\{k\in\mathbb{Z}|x\le k\}$.}
\]

We provide some initial results on the structure of the game. Let $m(S)$ be the optimal number of machines for coalition $S\subseteq N$.\footnote{There might be a tie, in which case pick the lowest number of machines among optimal ones.}

\begin{lemma}\label{alpha}
Fix the set of agents $N$. For any weight vector $w$ there exists a non-increasing function $r^w:\{2,...,n\}\to \mathbb{R}_+$ such that:
\begin{enumerate}[(i)]
\item if $b\geq r^w(2)$, then $m(N)=1$;
\item if $r^w(k)>b\geq r^w(k+1)$ for some $1<k<n$, then  $m(N)=k$; 
\item if   $r^w(n)>b$ then $m(N)=n$.
\end{enumerate}
\end{lemma}

We can similarly define a non-increasing function $r^w_S:\{2,...,|S|\}\to \mathbb{R}_+$ for all $S\subset N$ such that $|S|>1$ to determine $m(S)$. For singletons, it is always optimal to use a single machine, and thus $m(\{i\})=1$ for all $i\in N$. Observing the structure of these functions $r^w_S$, the following result follows:

\begin{lemma}\label{subset}
For all values of $w$ and $b$, we have:
\begin{enumerate}[(i)]
\item $m(S)\leq m(T)$ for all $S\subset T\subseteq N$;
\item $m(S\cup \{i\})\leq m(S\cup\{j\})$ for all $S\subseteq N\setminus \{i,j\}$, and $i>j$.
\end{enumerate}
\end{lemma}

In words, if we add agents to a coalition, it cannot be optimal to use less machines. The strategy to add an additional machine can only become more profitable (or less unprofitable) as the new agents might have higher waiting costs, and the additional agents might lead to more saved waiting costs. The same is true if we replace an agent by one with a larger waiting cost.
 
\section{Queueing games with an endogenous number of machines}
\label{sec:queue_games}
A \emph{cooperative transferable utility (TU) game} is defined by the pair $(N,C)$ where $N$ is the set of the players and the coalitional function $C$ assigns to each coalition $T\subseteq N$ its cost $C(T)\in\mathbb{R}$, with $C(\emptyset)=0$.

Cooperative game theory aims to allocate the value of the grand coalition in such a way that the cooperation is preserved among the agents. Given a cooperative game $(N,C)$, a \emph{cost allocation} is $y\in \mathbb{R}^{N}$, where $y_{i}$ stands for the cost paid by player $i\in N$. The total payment by a coalition $S\subseteq N$ is denoted by $y(S)=\sum\limits_{i\in S}y_{i}$ with $y(\emptyset)=0$. 

In this section, we study the set of stable allocations of the total cost, where no coalition of agents pays more than its stand-alone cost. To do so, for any queueing problem with an endogenous number of machines, we will introduce a TU-game and study the \emph{core} of the associated TU-game (\citealp{g59}).

Formally, let $(N,w,b)$ be a queueing problem with an endogenous number of machines. Then, the corresponding \emph{queueing game with an endogenous number of machines} is the pair $(N,C)$ where $N$ is the set of players, and $C$ is the characteristic function that assigns the minimal cost $C(T)$ to each coalition $T\subseteq N$ to queue its members, with $C(\emptyset)=0$. $C(T)$ includes both the waiting costs and the cost of machines.
The core of a cooperative cost game $(N,C)$ is: $$Core(C)=\{y\in\mathbb{R}^{N}\mid y(N)=C(N),\quad y(S)\leq C(S)\quad \text{for all}\quad S\subset N\}.$$ A game is called \emph{balanced} if its core is non-empty.

Concave cost functions always have a non-empty core (\citealp{s71}). Formally, a game $(N,C)$ is said to be \emph{concave} if for all $i\in N$ and all $S\subseteq T\subseteq N\setminus \{i\}$, it holds $C(T\cup\{i\})-C(T)\le C(S\cup\{i\})-C(S)$.

\subsection{On the non-emptiness of queueing games with an endogenous number of machines}

We look for conditions under the core is empty or non-empty. It turns out that for queueing games with an endogenous number of machines, the core can alternate between empty and non-empty. We first examine the cases when the cost of a machine is low, before examining the case when the cost is high. We conclude the section with an example illustrating Theorems \ref{thm:queue_game_upper_bound} and \ref{thm:queue_game_lower_bound} and how the core varies with the cost of machines.

First, for queueing games with an endogenous number of machines, we provide an upper bound on the cost of a machine to guarantee the non-emptiness of the core. 

For the sake of comprehensiveness, let us introduce some notation: Let $\mu\equiv \left\lceil \frac{n}{2}\right\rceil $. If $n$ is even, then $\{1,...,\mu\}$ and $\{\mu+1,...,n\}$ both contain $\mu$ agents, while if $n$ is odd, then $\{1,...,\mu\}$ contains $\mu$ agents and $\{\mu+1,...,n\}$ contains $\mu$-1 agents.

\begin{theorem}
\label{thm:queue_game_upper_bound}
Let $(N,w, b)$ be a queueing problem with an endogenous number of machines, and $(N,C)$ be the associated TU-game. 
\begin{enumerate}[(i)]
\item If $b\leq w_{\mu}$ then $y=\left( \min \left(b+w_{i},2w_i\right) \right) _{i\in N}\in Core(C)$;

\item If $b\leq w_{\left\lceil \frac{2n+1}{4}\right\rceil}$ then $Core(C)=\left( \min \left(b+w_{i},2w_i\right) \right) _{i\in N}$.
\end{enumerate}
\end{theorem}

\begin{proof}
(i) We first show that the allocation $y=\left( \min \left(b+w_{i},2w_i\right) \right) _{i\in N}$ is budget balanced. 

First, notice that if we use $k\geq \mu$ machines, than agents wait at most 2 periods. Adding an additional machine allows to reduce the waiting cost of agent $k+1$ from 2 to 1 period, with all other waiting costs remaining the same. Thus, $r^w(k)=w_k$ for all $k>\mu$. 

Notice also that when moving from $\mu-1$ to $\mu$ machines the cost savings are larger: in addition to agent $\mu$ waiting for 1 period instead of 2, some other agents will wait 2 periods instead of 3. Thus, $r^w(\mu)\geq w_{\mu}$.
Thus, given that $b\leq w_{\mu}\leq r^{w}(\mu)$ and by Lemma \ref{alpha}, $m(N)\geq \mu$.

Let $C(\cdot,k)$ be the cost function that assigns to each coalition the total cost if it uses $k$ machines to process their jobs. For $k\ge \mu$, $$C(N,k)=kb+\sum_{i=1}^kw_i+\sum\limits_{i=k+1}^{n}2w_{i}.$$

Thus, 
\begin{eqnarray*}
C(N) &=&\min_{k\in \left\{ \mu,...,n\right\} }\left\{
kb+\sum_{i=1}^{k}w_{i}+\sum_{i=k+1}^{n}2w_{i}\right\}  \\
&=&b(\mu-1)+\sum_{i=1}^{\mu-1}w_{i}+\min_{k\in \left\{ \mu ,...,n\right\} }\left\{b(k-\mu+1)+\sum_{i=\mu}^{k}w_{i}+\sum_{i=k+1}^{n}2w_{i}\right\}  \\
&=&\sum_{i=1}^{\mu-1}\left( b+w_{i}\right) +\sum_{i=\mu}^{n}\min
\left( b+w_{i},2w_{i}\right)  \\
&=&\sum_{i\in N}\min \left( b+w_{i},2w_{i}\right)  \\
&=&\sum_{i\in N}y_{i}
\end{eqnarray*}
The third equality comes from the fact that for all $k>\mu$, $r^{w}(k)=w_k$, implying that we use at least $k$ machines if and only $b+w_k\leq 2w_k$. While $r^w(\mu)\geq w_{\mu}$, by assumption $b\leq w_{\mu}$. The fourth equality also comes from the fact that by assumption, $b\leq w_{\mu}$.

It remains to prove that core constraints are satisfied, i.e., $y(T)\leq C(T)$ for all $T\subset N$. Fix $T\subset N$ and suppose that $\kappa$ is the optimal number of machines for $T.$

We have that 
\begin{eqnarray*}
\sum_{i\in T}y_{i} &\leq & \kappa b+\sum_{i=1}^{\kappa}w^T_i+\sum_{i=\kappa+1}^{\left\vert T\right\vert
}2w_{i}^{T}  \\
&\leq&C(T),
\end{eqnarray*}%
where the first inequality is obtained by assigning $b+w_i$ to the first $\kappa$ agents in $T$ and $2w_i$ to others, regardless of which of these two values is minimal, and the second inequality comes from the fact that the expression is exactly the cost of coalition $T$ if $\kappa\geq \frac{|T|}{2}$, with the cost no smaller otherwise. Thus the core constraint is satisfied. Since $T$ is arbitrarily chosen,
the proof is complete.

(ii) Notice first that if $n$ is odd, $\left\lceil \frac{2n+1}{4}\right\rceil=\mu$, while if $n$ is even, $\left\lceil \frac{2n+1}{4}\right\rceil=\mu+1$. In particular, for any $n$, we have that $\left\lceil \frac{2n+1}{4}\right\rceil-1\geq \frac{n-1}{2}$.

If $b< w_n$, then $C(S)=|S|b+\sum_{i\in S} w_i$ for all $S\subseteq N$ and the result is immediate. Thus, suppose that $b\geq w_n$.

Suppose that $w_{k+1}\leq b < w_k$ for $k\in \left\{\left\lceil \frac{2n+1}{4}\right\rceil,...,n-1\right\}$. Then, by Lemma \ref{alpha}, $C(N)=kb+\sum_{i=1}^kw_i+\sum_{i=k+1}^n2w_i$.

Consider coalition $N\setminus \{i\}$ for $i\in \left\{1,...,k\right\}$. If they use $k$ machines, the cost is $kb+\sum_{j=1}^kw_i+\sum_{j=k+1}^n2w_i-w_i-w_{k+1}$. If they use $k-1$ machines, the cost is $(k-1)b+\sum_{j=1}^kw_i+\sum_{j=k+1}^n2w_i-w_i$, as $k-1\geq \left\lceil \frac{2n+1}{4}\right\rceil-1\geq \frac{n-1}{2}$. Thus, it prefers to use $k-1$ machines if $b\geq w_{k+1}$, which is satisfied. If they use $k-2$ machines, the cost is at least $(k-2)b+\sum_{j=1}^kw_i+\sum_{j=k+1}^n2w_i-w_i+w_{k}$ (as some agents might have to wait more than 2 periods now), and as $b\leq w_k$ it prefers to use $k-1$ machines. Thus, $C(N\setminus \{i\})=(k-1)b+\sum_{j=1}^kw_i+\sum_{j=k+1}^n2w_i-w_i$.

Notice that $C(N\setminus \{i\})+C( \{i\})=C(N)$, and thus in any core allocation, we must have $y_i=C( \{i\})=b+w_i$ for all $i\in \left\{1,...,k\right\}$.

Next, consider coalition $\{i,j\}$, with $i\in \left\{1,...,k\right\}$ and $j\in \left\{k+1,...,n\right\}$. If it uses a single machine, the cost is $b+w_i+2w_j$. If it uses 2 machines, the cost is $2b+w_i+w_j$. It prefers to use a single machine as $b\geq w_{k+1}\geq w_j$. Thus, $C(\{i,j\})=b+w_i+2w_j$. Since $y_i=b+w_i$, we obtain a core constraint of $y_j\leq 2w_j$ for all $j\in \left\{k+1,...,n\right\}$. Given the value of $C(N)$, our only core candidate is  $y_i=b+w_i$ for all $i\in \left\{1,...,k\right\}$ and $y_j=2w_j$ for all $j\in \left\{k+1,...,n\right\}$. Given that we have shown in part i) that it is a core allocation, our proof is complete.
\end{proof}

Following an upper-bound on the cost of a machine for the non-emptiness of the core, we provide a full characterization of the core making use of a lower-bound for the non-emptiness of the core. 

\begin{theorem}
\label{thm:queue_game_lower_bound}
Let $(N,w,b)$ be a queueing problem with an endogenous number of machines, and $(N,C)$ be the associated TU-game. Then,
\begin{enumerate}[(i)]
    \item if $b\ge \sum_{i=1}^{n}(i-1)w_{i}$, $Core(C)=Core(\hat{C})\neq\emptyset$ with $\hat{C}(T):=C(T)-\sum\limits_{i=1}^{n-|T|-1}iw_{i+1}^{-T}$ for all $\emptyset\neq T\subseteq N$. Moreover, $\hat{C}$ is concave.
    \item if $b\in \left[w_2+\sum\limits_{i=3}^{n}\left(i-\left\lceil\frac{i}{2}\right\rceil\right)w_{i},\sum_{i=1}^{n}(i-1)w_{i} \right)$, then $Core(C)= \emptyset$.
\end{enumerate}
\end{theorem}

\begin{proof}

Notice that $w_2+\sum\limits_{i=3}^{n}\left(i-\left\lceil\frac{i}{2}\right\rceil\right)w_{i}=r^w(2)$, and thus by Lemma \ref{alpha}, $m(N)=1$. By Lemma \ref{subset}, $m(S)=1$ for all $S\subseteq N$, and thus all coalitions use a single machine.

i) Notice first that $\hat{C}(T)=C(T)$ if $\left\vert T\right\vert \geq n-1.$ Let $j,k\in N$ and consider $N\backslash \left\{ j,k\right\} .$ By the grand coalition efficiency and individual rationality, if $y$ is a core allocation, then $y(N\setminus\left\{ j,k\right\} )\leq C(N\setminus \left\{ j\right\} )+C(N\setminus\left\{ k\right\} )-C(N).$

Suppose that we have shown that $y(T)\leq \hat{C}(T) $ in any core allocation if $\left\vert T\right\vert >m.$ We need to show that it implies that $y(T)\leq \hat{C}(T)$ in any core allocation if $\left\vert T\right\vert =m.$

Fix $T$ such that $\left\vert T\right\vert =m.$ Following the grand coalition efficiency and individual rationality, if $y$ is a core allocation, then $y(T)\leq \hat{C}(T\cup \left\{
k\right\} )+\hat{C}(N\backslash \left\{ k\right\} )-\hat{C}(N).$

Now, consider the cost function $\hat{C}(T)=C(T)-\sum_{i=1}^{n-\left\vert T\right\vert
-1}iw_{i+1}^{-T},$ for all $\emptyset \neq T\subseteq N.$ By definition, $%
\hat{C}\leq C$ and $\hat{C}(N)=C(N).$ We will show that $\hat{C}$ is concave whenever the lower-bound on the cost of a machine in ($\emph{i}$) is satisfied.

First, notice that for all $k\in N$, $\hat{C}%
(\left\{ k\right\} )=b+w_k-\sum_{i=1}^{n-2}iw_{i+1}^{-\left\{ k\right\}
}=b+w_k-\sum_{i<k}(i-1)w_{i}-\sum_{i>k}(i-2)w_{i}.$

Recall that $w^{T}_{k}$ denotes the waiting cost of the $k^{th}$ agent in $T$, according to the order in $N$ and $w^{-T}_{k}\equiv w^{N\setminus T}_{k}$. Next, fix $\emptyset \neq T\subseteq N\backslash \left\{ k\right\}.$ Then, we have that 
\begin{eqnarray*}
\hat{C}(T\cup \left\{ k\right\} )-\hat{C}(T) &=&\sum_{i=1}^{\left\vert
T\right\vert+1 }iw_{i}^{T\cup \left\{ k\right\} }-\sum_{i=1}^{\left\vert
T\right\vert }iw_{i}^{T}-\sum_{i=1}^{n-\left\vert T\right\vert
-2}iw_{i+1}^{-\left( T\cup \left\{ k\right\} \right)
}+\sum_{i=1}^{n-\left\vert T\right\vert -1}iw_{i+1}^{-T} \\
&=&\sum_{i>k}w_{i}+kw_{k}
\end{eqnarray*}

The equality is based on the following observations: if $i<k$ and $i\in T$, then its rank in $T\cup \left\{ k\right\}$ is the same as in $T,$ and the terms cancel out. The same is true if $i\in N\setminus T.$ If $i>k$ and $i\in T,$ the rank of $i$ is one higher in $T\cup \left\{ k\right\}$ than in $T.$ If $i>k$ and $i\in N\backslash T$, the rank of $i$ is one smaller in $N\setminus \left( T\cup \left\{ k\right\} \right) $ than in $N\setminus T.$ In all cases, the difference is $w_{i}.$ As for $k,$ it appears in the first and fourth terms. The weight on its waiting cost is its rank in $T\cup \left\{k\right\} $ plus its rank in $N\setminus T$ minus 1$.$ For all agents, that
equals $k$.

This result is independent of $T,$ as long as $T\neq \emptyset .$ Making use of this result, we proceed to show $y(T)\le \hat{C}(T).$  We have seen that $\hat{C}(T\cup \left\{ k\right\} )= b+\sum_{i=1}^{\left\vert T\right\vert +1 }iw_{i}^{T\cup \left\{
k\right\} }-\sum_{i=1}^{n-\left\vert T\right\vert -2}iw_{i+1}^{-(T\cup
\left\{ k\right\} )}$ and $\hat{C}(N\backslash \left\{ k\right\} )= b+\sum_{i=1}^{n-1}iw_{i}^{N\backslash \left\{ k\right\} }$and $%
\hat{C}(N)=b+\sum_{i=1}^{n}iw_{i}.$ Thus,
\[
\hat{C}(T\cup \{k\})+\hat{C}(N\setminus\left\{k\right\})-\hat{C}(N)=b+\sum\limits_{i=1}^{|T|+1}iw_{i}^{T\cup\left\{k\right\}}-\sum\limits_{i=1}^{n-|T|-2}iw_{i+1}^{-(T\cup\left\{k\right\})}+\sum\limits_{i=1}^{n-1}iw_{i}^{N\setminus \left\{k\right\}}-\sum\limits_{i=1}^{n}iw_{i}.
\]

Next, let us distinguish several cases.

\noindent\textit{Case 1: $i\in T$ such that $i<k$.} Then, $i$ is the $i^{th}$ agent in $N\backslash
\left\{ k\right\} $ and in $N.$ The agent $i$ has the same rank in $T$ and in $T\cup
\left\{ k\right\} .$

\noindent\textit{Case 2: $i\in T$ such that $i>k.$} Then, $i$ is the $(i-1)^{th}$ agent in $%
N\backslash \left\{ k\right\} $ and $i^{th}$ in $N.$ The agent $i$ has the same rank in $T$ is one
lower than in $T\cup \left\{ k\right\} .$

\noindent\textit{Case 3: $j\notin T$ such that $j<k.$} Then, $j$ is the $j^{th}$ agent in $%
N\backslash \left\{ k\right\} $ and in $N.$ The agent $j$ has the same rank in $%
N\backslash T$ and in $N\backslash (T\cup \left\{ k\right\} ).$

\noindent\textit{Case 4: $j\notin T$ such that $j>k.$} Then, $j$ is the $(j-1)^{th}$ agent in $%
N\backslash \left\{ k\right\} $ and $j^{th}$ in $N.$ The rank of $j$ in $N\setminus T$ is one higher than in $N\backslash (T\cup \left\{ k\right\}
). $

Note that in the second term agent $k$ his waiting cost gets assigned a weight equal to his rank in $T\cup \{k\}$. In the fifth term, his waiting cost gets assigned a weight equal to $k$, his rank in $N$. These cases imply that 
\begin{eqnarray*}
\hat{C}(T\cup \left\{ k\right\} )+\hat{C}(N\backslash \left\{ k\right\} )-%
\hat{C}(N)
&=&b+\sum_{i=1}^{\left\vert T\right\vert
}iw_{i}^{T}-\sum_{i=1}^{n-\left\vert T\right\vert -1}iw_{i+1}^{-T}
\\
&=&\hat{C}\left(T \right).
\end{eqnarray*}

Thus, we have the core constraint $y(T)\leq \hat{C}(T)$. In order to finish the proof of $(\emph{i})$, we will show that $\hat{C}$ is concave. 

To verify concavity, it remains to check that $\hat{C}(T\cup \left\{ k\right\}
)-\hat{C}(T)\leq \hat{C}\left( \left\{ k\right\} \right) -\hat{C}\left(
\emptyset \right) =\hat{C}\left( \left\{ k\right\} \right)$ or equivalently,
\[
\sum_{i>k}w_{i}+kw_{k}\leq b+w_k-\sum_{i<k}(i-1)w_{i}-\sum_{i>k}(i-2)w_{i} 
\]%
which simplifies to
\[
b\geq \sum_{i=1}^{n}(i-1)w_{i} 
\]%
which is satisfied by assumption. Thus, $\hat{C}$ is concave and hence $Core(\hat{C})$ is non-empty which finishes the proof of ($\emph{i}$).

To prove the statement ($\emph{ii}$), recall that we have shown that $C(N\backslash \left\{ l\right\} )=b+\sum_{i=1}^{n-1}iw_{i}^{N\backslash \left\{ l\right\} }$ for $%
l\in N$ and $C(N)= b+\sum_{i=1}^{n}iw_{i}.$ Thus, 
\begin{eqnarray*}
\sum_{l\in N}C\left( N\backslash \left\{ l\right\} \right)  &=&\sum_{l\in
N}\left( b+\sum_{i=1}^{n-1}iw_{i}^{N\backslash \left\{ l\right\} }\right)  \\
&=&nb+\sum_{l\in N}\sum_{i=1}^{n-1}iw_{i}^{N\backslash \left\{ l\right\} } \\
&=&nb+\sum_{i\in N}\left( i\left( n-2\right) +1\right) w_{i}.
\end{eqnarray*}

The last equality is obtained as follows: agent $i$ does not appear in $%
N\backslash \left\{ i\right\} ,$ appears at rank $i$ in $N\backslash \left\{
l\right\} $ if $l>i,$ and appears at rank $i-1$ if $l<i.$ We have $n-i$
coalitions where $l>i$ and $i-1$ coalitions where $l<i.$ Thus, the
coefficient associated to $w_{i}$ is $(n-i)i+(i-1)(i-1)=i(n-2)+1.$

Then, we have that 
\begin{eqnarray*}
\sum_{l\in N}C\left( N\backslash \left\{ l\right\} \right) -(n-1)C(N)
&=&nb+\sum_{i\in N}\left( i\left( n-2\right) +1\right) w_{i}-\left( \left(
n-1\right) b+\sum_{i\in N}\left( n-1\right) iw_{i}\right)  \\
&=&b-\sum_{i\in N}(i-1)w_{i} \\
&<&0,
\end{eqnarray*}
and thus $\sum_{l\in N}C\left( N\backslash \left\{ l\right\} \right)
<(n-1)C(N),$ by our assumption on $b.$ Thus, $Core(C)$ is empty.

\end{proof}

We provide an example that shows that the conditions in Theorems \ref{thm:queue_game_upper_bound} and \ref{thm:queue_game_lower_bound}i) are not necessary for the core to be non-empty.

\begin{example}
\label{ex:core}
Suppose that $N=\{1,2,3,4\}$ and that $w_i=25-5i$ for all $i\in N$. Then, Theorem \ref{thm:queue_game_upper_bound} tells us that the core is non-empty if $b\leq 15$ while Theorem \ref{thm:queue_game_lower_bound}i) guarantees non-emptiness of the core for $b\geq 50$. For $b\in \left[35, 50\right)$, by Theorem \ref{thm:queue_game_lower_bound}ii) the core is empty. We verify what happens when $b\in (15,35)$. 

For $b\in \left[20,35\right)$, coalitions $\left\{1,2,3\right\}$ and $\left\{1,2,3,4\right\}$ use 2 machines, all others use a single one. Using the fact that we must have $y\left(\{i,j\}\right)=b+w_i+2w_j$ for all $i\in \{1,2\}, j\in \{3,4\}$ in any core allocation, we obtain maximal allocations of $(40,35,25,15)$. Using $C(N)-C(N\setminus \{i\})$, we obtain minimal allocations of $(b+15, b+10, b, b-10)$. Immediately, the core is empty for $b\in (25, 35)$. For $b\in \left[20,25\right]$, we can verify that the allocation $(b+15,b+10,25,15)$ is in the core.

For $b\in (15,20)$, coalition $\left\{1,2,4\right\}$ also uses 2 machines. Using the same technique as above to obtain minimal and maximal allocations, our only candidate for a core allocation is $(b+15,b+10,25,15)$. But then, coalition $\{3,4\}$ pays 40, while its stand-alone cost is $b+20$, and thus there are no core allocations. 

For $b\in\left(10,15 \right]$, all coalitions of 3 or more agents use 2 machines, as well as coalition $\left\{1,2\right\}$. In addition to $\left(b+20, b+15, 20, 10\right)$, the allocation $\left(\frac{b}{2}+25, \frac{b}{2}+20, \frac{b}{2}+15, \frac{b}{2}+5\right)$ is also in the core. 

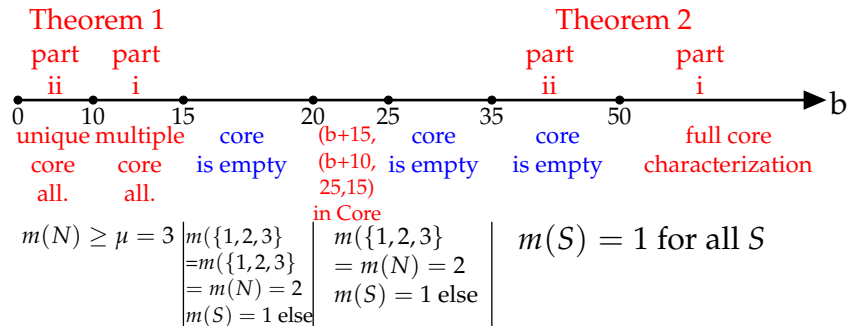
\begin{figure}[!h]
\hspace{1.5cm}
\label{fig:example1fig}
\begin{tikzpicture}[line cap=round,line join=round,>=triangle 45,x=0.5cm,y=0.5cm]
\clip(3.8,-2) rectangle (26.5,7);
\draw [->,line width=1.pt] (4.,4.) -- (25.5,4.);
\draw (25.3,4.535) node[anchor=north west] {b};
\draw (4,6.75) node[anchor=north west] {{\color{red}{\small{Theorem 1}}}};
\draw (18,6.75) node[anchor=north west] {{\color{red}{\small{Theorem 2}}}};
\draw (5.85,5.9) node[anchor=north west] {{\footnotesize{\begin{tabular}{c} {\color{red}{part}} \\ {\color{red}{i}} \end{tabular}}}};
\draw (16.85,5.9) node[anchor=north west] {{\footnotesize{\begin{tabular}{c} {\color{red}{part}} \\ {\color{red}{ii}} \end{tabular}}}};
\draw (20.85,5.9) node[anchor=north west] {{\footnotesize{\begin{tabular}{c} {\color{red}{part}} \\ {\color{red}{i}} \end{tabular}}}};
\draw (3.7,5.9) node[anchor=north west] {{\footnotesize{\begin{tabular}{c} {\color{red}{part}} \\ {\color{red}{ii}} \end{tabular}}}};
\draw (3.3,3.7) node[anchor=north west] {{\scriptsize{\begin{tabular}{c} {\color{red}{unique}} \\ {\color{red}{core}} \\ {\color{red}{all.}} \end{tabular}}}};
\draw (5.4,3.7) node[anchor=north west] {{\scriptsize{\begin{tabular}{c} {\color{red}{multiple}} \\ {\color{red}{core}} \\ {\color{red}{all.}} \end{tabular}}}};
\draw (8.1,3.7) node[anchor=north west] {{\scriptsize{\begin{tabular}{c} {\color{blue}{core}} \\ {\color{blue}{is empty}} \end{tabular}}}};
\draw (11.2,3.7) node[anchor=north west] {{\tiny{\begin{tabular}{c} {\color{red}{(b+15,}} \\ {\color{red}{(b+10,}}   \\ {\color{red}{25,15)}} \\ {\color{red}{in Core}} \end{tabular}}}};
\draw (13.2,3.7) node[anchor=north west] {{\scriptsize{\begin{tabular}{c} {\color{blue}{core}} \\ {\color{blue}{is empty}} \end{tabular}}}};
\draw (16.5,3.7) node[anchor=north west] {{\scriptsize{\begin{tabular}{c} {\color{blue}{core}} \\ {\color{blue}{is empty}} \end{tabular}}}};
\draw (20,3.7) node[anchor=north west] {{\scriptsize{\begin{tabular}{c} {\color{red}{full core}} \\ {\color{red}{characterization}} \end{tabular}}}};

\draw (3.8,1) node[anchor=north west] {\scriptsize{$m(N) \ge \mu =3$}};
\draw (7.8,1) node[anchor=north west] {\tiny{\begin{tabular}{l} {$m(\{1,2,3\}$} \\ {=$m(\{1,2,3\}$} \\ {$=m(N)=2$} \\ {$m(S)=1$ else} \end{tabular}}};
\draw (11.75,1) node[anchor=north west] {\scriptsize{\begin{tabular}{l} {$m(\{1,2,3\}$}  \\ {$=m(N)=2$} \\ {$m(S)=1$ else} \end{tabular}}};
\draw (17,1) node[anchor=north west] {$m(S)=1$ for all $S$};

\draw [line width=0.5pt] (8.4,0.75)-- (8.4,-2);
\draw [line width=0.5pt] (11.85,0.75)-- (11.85,-2);
\draw [line width=0.5pt] (16.6,0.75)-- (16.6,-2);
\begin{scriptsize}
\draw [fill=black] (4.,4.) circle (1.5pt);
\draw[color=black] (4,3.575) node {0};
\draw [fill=black] (6.,4.) circle (1.5pt);
\draw[color=black] (6,3.575) node {10};
\draw [fill=black] (8.4,4.) circle (1.5pt);
\draw[color=black] (8.4,3.575) node {15};
\draw [fill=black] (11.85,4.) circle (1.5pt);
\draw[color=black] (11.85,3.575) node {20};
\draw [fill=black] (13.85,4.) circle (1.5pt);
\draw[color=black] (13.85,3.575) node {25};
\draw [fill=black] (16.6,4.) circle (1.5pt);
\draw[color=black] (16.6,3.575) node {35};
\draw [fill=black] (20,4.) circle (1.5pt);
\draw[color=black] (20,3.575) node {50};
\end{scriptsize}
\end{tikzpicture}
\caption{Summary of Example \ref{ex:core}.}
\end{figure}


The example is summarized in Figure \ref{fig:example1fig}, with the results on the core and the description of $m$, the optimal number of machines, depending on machine cost $b$. In order to provide a clear illustration, Figure \ref{fig:example1fig} does not respect the appropriate proportions.
\end{example}

When $b\leq w_{\mu}$, the game we obtain is reminiscent of assignments games: we need to match agents with a machine (those with $b\leq w_i$) to those without (with $b>w_i)$, matching at most one agent from the second group to each agent in the first group. An agent from the second group generates value of $b-w_i$ when he matches with any agent from the second group. Notice that the first group contains at least half of the agents. If it is strictly more, the only core allocation allocates all gains to the short side of the market, the agents with $b>w_i$. If it is exactly half, as in our example above with $b\in\left(10,15 \right]$, the core contains multiple allocations. This is where the similarities with the assignment game ends, as members of the second group do create value when matched together (they share a machine), which constrains the allocation that is optimal for the first side (those with $b\leq w_i$).

\section{Requeueing games with an endogenous number of machines}
\label{sec:requeue}
For the rest of the paper, we consider queueing problems with an existing queue, which study the problem from a different perspective: while queueing problems consider the minimal cost of organizing the queue for a set of players, starting from scratch, in the following, we consider requeueing problems where possible cost savings can be obtained when we reschedule a given queue. In our study of the problem with an endogeneous number of machines, this implies that we start with a given number of machines, and that the reorganization can include adding or removing machines.\footnote{We use queueing problems with an existing queue and requeuing problems interchangeably.}

Then, a \emph{requeueing problem with an endogenous number of machines} can be described by $(N,m_{0},\sigma_{0},w,b)$ where $m_{0}$ is the initial number of machines and $\sigma_{0}$ is the initial (existing) queue. Our first aim is to find an optimal schedule that minimizes the total costs, as in Section \ref{sec:queue}. As for queueing games, we build a coalitional function from the requeueing games, now associating to each coalition $T\subseteq N$ the maximum cost savings $V(T)$ it can generate from the initially existing queue. We will distinguish between two cases based on whether new machines are exclusive for a set of agents (private) or available for all agents (public). 

\subsection{Private requeueing games}

We consider requeueing problems in which if a coalition buys a new machine, it gains exclusive use of that machine and if a coalition sells a machine it recovers the full value of that machine. These two assumptions can be seen as ``exclusive'' use of machines for a coalition and hence they are ``private'' machines for a coalition. 



In order to determine the maximal cost savings of a coalition $T\subseteq N$, we have to define which rearrangements are admissible. Various assumptions have been made on admissible rearrangements of the initial schedule, see \cite{cetal93}, \cite{s06b}, \cite{metal15}, and \cite{aetal21}, among others. Following the literature, we consider two approaches to define admissible rearrangements to study the non-emptiness of the core for requeuing games with an endogenous number of machines.



First, we do not allow agents in a coalition to jump over agents outside the coalition. Then, we say that a scheduling plan $(m,\sigma$) is \emph{admissible} for a coalition $T$ with respect to $(m_0,\sigma_{0})$ if for any agent outside coalition $T$ there are no new agents in her set of predecessors. That is,
for all $i\in N\setminus T$ it holds that $\varphi(i)=\varphi_0(i)$ and
\begin{equation}
\{l\in N^{\varphi(i)}:s(l)<s (i)\}\subseteq\{l\in N^{\varphi_0(i)}:s_{0}(l)<s_0(i)\}.\label{adm_sch_1} 
\end{equation}
Notice that we do not require equality, as a predecessor of $i$ might move to a new machine. For short, we call this assumption the ''no swaps'' assumption, and the set of admissible schedules for coalition $T$ that satisfy (\ref{adm_sch_1}) is denoted by $\Sigma^{ns}_{T}$.

Second, we relax the condition by allowing agents in a coalition to jump over agents outside the coalition. We say that a scheduling plan $(m,\sigma$) is \emph{admissible} for a coalition $T$ with respect to $(m_0,\sigma_{0})$ if the starting time for all agents outside the coalition does not increase. That is, for all $i\in N\setminus T$ it holds that $\varphi(i)=\varphi_0(i)$ and
\begin{equation}
s(i)\leq s_0(i). \label{adm_sch_2}
\end{equation}
Once again, we do not have equality, as predecessors are allowed to move to a new machine. By opposition, this is the ''swaps'' assumption, and the set of admissible schedules for coalition $T$ that satisfy (\ref{adm_sch_2}) is denoted by $\Sigma^{s}_{T}$.

In our setting, we also must consider the possibility for a coalition to sell a machine. We suppose that a coalition $T$ can sell a machine only if all users of that machine are members of $T$. We then suppose that the agents that were on the removed machine move at the end of the queue on the remaining machines, a condition that is already covered by both (\ref{adm_sch_1}) and (\ref{adm_sch_2}).

For a set of admissible schedules, we can associate the corresponding cooperative TU-game called a private requeueing game with an endogenous number of machines. 
A \emph{private requeuing problem with an endogenous number of machines} is a 5--tuple $(N,m_{0},\sigma_{0},w,b)$.\footnote{Since population and costs are fixed, with an abuse of notation, we denote a private requeuing problem also by the initial number of machines and the initial queue, $(m_0,\sigma_{0})$.} The corresponding \emph{private requeueing game with an endogenous number of machines} $(N,V)$ is defined by $$V(T)=c_{\sigma_{0}}(T)-c_{\sigma}(T)-(m-m_0)b,$$ where $(m, \sigma)$ is an optimal admissible scheduling plan for coalition $T$. 
Furthermore, admissible schedules with and without swaps lead to different games. We denote private requeueing games with swaps by $V_{s}$ and private requeueing games without swaps by $V_{ns}$.

While concavity of a cost game is a sufficient condition for its core to be non-empty, for value games the corresponding concept is that of convexity. The condition has been widely studied to prove balancedness of sequencing games associated with different problems (see for instance \citealp{cetal94}; \citealp{hetal05}; \citealp{metal18}). A game $(N,v)$ is said to be \emph{convex} if for all $i\in N$ and all $S\subseteq T\subseteq N\setminus \{i\}$, it holds $v(T\cup\{i\})-v(T)\ge v(S\cup\{i\})-v(S)$. 

Unfortunately, Example \ref{ex:hurt} shows that the associated game need not be balanced, regardless if admissible schedules allow or not to jump over players outside the coalition.

\begin{example}
\label{ex:hurt}
Consider $(N,m_0,\sigma_0,w,b)$ with $N=\{1,2,3,4,5\}$. The waiting costs per unit for agents are given by the weight vector $w=(w_{i})_{i\in N}=(20,15,13,13,5)$, and the cost of a machine is $b=18$.

First, we suppose that $(m_0,\sigma_0)$ is such that we order agents in the queue on one machine according to their weights, in decreasing order:

\begin{center}
\begin{tabular}{c|c|c|c|c|c|}
\cline{2-6}
$m_{1}$ & 1 & 2 & 3 & 4 & 5 \\ \cline{2-6}
\end{tabular}.%
\end{center}
Notice that agent 1 is a dummy player since he is served first, moving to another machine is strictly worse for her. Thus, we can focus on the game $(N,V_s)$ for the remaining agents. One can calculate that $V_{s}(\{2,3,4\})=36$, $V_{s}(\{2,3,5\})=25$, $V_{s}(\{3,4,5\})=44$, and $V_{s}(\{2,3,4,5\})=46$. Next, let us consider the coalition $T=\{2,4,5\}$.

First, suppose that we allow players in the coalition to jump over players outside the coalition when we define admissible rearrangements. Take the coalition $T=\{2,4,5\}$. Then, an optimal scheduling plan for coalition $T$, $(m_T, \sigma_T)$, is
\begin{center}
\begin{tabular}{c|c|c|c|}
\cline{2-4}
$m_{1}$ & 1 & \textbf{4} & 3  \\ \cline{2-4}
$m_{2}$ & \textbf{2} & \textbf{5} & \\ \cline{2-4}
\end{tabular},
\end{center}
and then the total waiting cost savings are 15+26+15=56, but the coalition buys a machine at a cost 18, and hence the maximal total cost savings is $38=V_{s}(\{2,4,5\})$. Then, $V_{s}(\{2,3,4\})+V_{s}(\{2,3,5\})+V_{s}(\{3,4,5\}+V_{s}(\{2,4,5\})=143>138=3V_{s}(\{2,3,4,5\})$, and hence the core is empty.

Second, suppose that we do not allow players in the coalition to jump over players outside the coalition when we define admissible rearrangements. Take the coalition $T=\{2,4,5\}$. Then, an optimal scheduling plan for coalition $T$, $(m'_T, \sigma'_T)$, is
\begin{center}
\begin{tabular}{c|c|c|c|}
\cline{2-4}
$m_{1}$ & 1 & \textbf{2} & 3  \\ \cline{2-4}
$m_{2}$ & \textbf{4} & \textbf{5} & \\ \cline{2-4}
\end{tabular},
\end{center}
and then the total waiting cost savings are 39+15=54, but the coalition buys a machine at a cost 18, and hence the maximal total cost savings is $36=V_{ns}(\{2,4,5\})$. For all other coalitions $S$, we have $V_{ns}(S)=V_{s}(S)$. Then, $V_{ns}(\{2,3,4\})+V_{ns}(\{2,3,5\})+V_{ns}(\{3,4,5\})+V_{ns}(\{2,4,5\})=141>138=3V_{ns}(\{2,3,4,5\})$, and hence the core is empty.
\end{example} 

Although we have seen that the associated private requeueing game can have an empty core, there are sufficient conditions to guarantee the non-vacuity of the core. First, we provide a lower bound on the cost of a machine for the non-emptiness of the core.

\begin{theorem}
\label{thm:lower_bound_sequencing}
Let $(N,m_0,\sigma_0,w,b)$ be a private requeuing problem and let $(N,V_{ns})$ be the associated private requeuing game without swaps, and $(N,V_{s})$ be the associated private requeueing game with swaps. If $b\le w_{n}$, then $Core(V_{ns})\neq\emptyset$ and $Core(V_{s})\neq\emptyset$.
\end{theorem}

\begin{proof}
Since the cost a machine is at most equal to the smallest unit waiting cost, $b\le w_{n}$, for any coalition $S\subseteq N$, it is optimal to have $|S|$ machines. Let $s_0(i)$ be the starting time of the job $i$ under the schedule $\sigma_0$. The first agents at each machine in the initial schedule, that is $i\in N$ such that $s_0(i)=0$, need not change their position. For all $S\subseteq N$, let $S_0$ be the set of such agents. For all other agents $i\in N\setminus N_0$ such that $s_0(i)\ge 1$, buying a new machine is the option that maximizes the cost savings at any given coalition since $b\le w_{n}$. Thus there exists a unique core allocation where $y_{i}=s_0(i)w_{i} -b$ for all $i\in N\setminus N_0$ and $y_{i}=0$ for $i\in N_0$. Then, $\sum\limits_{i\in N} y_{i}= \sum\limits_{i\in N_0} 0 + \sum\limits_{i\in N\setminus N_0}(s_0(i)w_{i} - b)=\sum\limits_{i\in N\setminus N_0}s_0(i)w_{i} -(n-m_0)b= V_{s}(N)=V_{ns}(N)$, and efficiency holds. For any coalition $S\subset N$, $y(S)=\sum\limits_{i\in S\setminus S_{0} }(s_0(i)w_{i} - b)$. Since using a machine for each agent in the coalition is the optimal schedule in both cases $y(S)=V_{s}(S)=V_{ns}(S)$. Thus, $y$ also satisfies coalitional rationality and hence it is a core allocation.
\end{proof}

Following a lower-bound on the cost of a machine to guarantee the non-emptiness of the core, we provide an upper-bound on the cost of a machine for the non-emptiness of the core when the initial number of machines is 1, $m_{0}=1$. Intuitively, it consists in setting the machine cost so high that no coalition wants to buy a second machine. The problem then becomes one or reorganizing the queue on the existing machine, making the problem equivalent to one with a single machine and no possibility to add more. 


\begin{theorem}
\label{thm:upper_bound_sequencing}
Let $(N,m_0,\sigma_0,w,b)$ be a private requeuing problem. If $m_{0}=1$ and $b\ge\max_{k=1}^{\mu}(n-k)w_{i}$, then $Core(V_{s})\neq\emptyset$ and $Core(V_{ns})\neq\emptyset$.\end{theorem}

\begin{proof}
We will show that for the private sequencing game without swaps where the agents are ordered in an increasing way with respect to their waiting costs at the initial order, making use of the only machine is better than buying a new machine for any coalition. Given that this is the worst case scenario in terms of ordering, and we still do not want to buy more than one machine, the result holds for all other orderings. Formally, we consider an initial schedule $\sigma_0$ such that $s_0(i)=n-i$ for all $i\in N.$

First, consider the last agent in the order, agent 1, which by definition is such that $w_1 \ge (w_i)_{i\in N\setminus \{1\}}$. If she buys a machine, the total cost savings are $(n-1)w_{1}-b$. Since $b\ge\max_{k=1}^{\mu}(n-k)w_{i}\ge (n-1)w_{1}$, the total cost savings $(n-1)w_{1}-b\le 0$. Hence, agent 1 is worse off by buying a new machine. Note that for any other individual coalition $\{i\}$ such that $i\in N\setminus \{1\}$, since the gain by buying a new machine is $(n-i)w_{i}<(n-1)w_{1}$, the total cost savings $(n-i)w_{i}-b<(n-1)w_{1}-b\le 0$. Thus, no individual coalition $\{i\}$ such that $i\in N$ has incentives to buy a new machine.


We next show that it is true for a coalition containing $k\leq \mu$ agents. Consider the coalition of the last $k$ agents in the order, agents 1 through $k$. Recall that $w_{1}\ge\ldots\ge w_{k}\ge w_{i}$ for all other $i\in N$ such that $s(i)<n-k$. Then, the gain for coalition $\{1,\ldots, k\}$ is $(n-1)w_{1}+\ldots +(n-k)w_{k}$. Since  $b\ge\max_{k=1}^{\mu}(n-k)w_{i}\ge (n-1)w_{1}+\ldots +(n-k)w_{k}$, the total cost savings for $\{k-1,\ldots, 1\}$ is $(n-1)w_{1}+\ldots +(n-k)w_{k}-b\le 0$, and hence coalition $\{1,\ldots, k\}$ prefers to use one machine. Notice that no other coalition $S$ with $|S|=k$ can achieve higher total cost savings than coalition $\{1,\ldots, k\}$, hence no $k$-agent coalition has incentives to buy a new machine.


Notice that for any coalition that consists of last $k$ agents where $k>\mu$, only agents with a position in the initial order of more than $\mu$ get to use the new machine. Hence, their gains are less than the last $\mu$ agents, and they also prefer to use the only machine provided to them. Moreover, since we compare all possible coalitions with the same size coalition that consists of the last agents in the initial queue, our result under not allowing swaps subsumes if we allow swaps. Then, we deal with a 1-machine problem with an initial queue. Together with the result of \cite{cetal89} and \cite{hetal99} stating that the core is always non-empty for 1-machine problem, we guarantee the non-emptiness of the core whenever $b\ge\max_{k=1}^{\mu}(n-k)w_{i}$ with $m_{0}=1$.
\end{proof}

Given that the previous result is obtained using the worst case scenario of a completely inefficient initial ordering, for a random initial ordering a lower bound guaranteeing a non-empty core could be found. However, a general expression for such a bound is difficult to obtain. 

\subsection{Public requeueing games}

Implicit in the previous subsection was the assumption that if a coalition buys a new machine it would gain exclusive use of that machine. That does not have to be the case. We consider here the opposite assumption, in which new machines are available for all agents. To illustrate the differences between the two assumptions, suppose that we have a long queue of agents waiting to go through security/ticket control at a sporting event. If somehow agents waiting got hold of an additional employee who could, given appropriate compensation, open a new lane to speed up the process, who would have access to that lane? Up to now, we had supposed that this new lane would be a VIP lane, accessible only to agents who helped compensate this additional worker. But another reasonable interpretation is that this new lane would be available to all, making this new machine a public good.

We illustrate by returning to Example \ref{ex:hurt}: if a new line opens up, the initial schedule $\sigma_{0}$ is split up in two: 1 and 2 are served first, 3 and 4 second, and 5 third and the new scheduling plan $(m',\sigma^{\prime})$ is
\begin{center}
\begin{tabular}{c|c|c|c|}
\cline{2-4}
$m_{1}$ & 1 & 3 & 5  \\ \cline{2-4}
$m_{2}$ & 2 & 4 & \\ \cline{2-4}
\end{tabular}.
\end{center}

Then, we can calculate the worth of the coalition $T=\{2,4,5\}$, as the waiting costs saved by its members only, net of the new machine cost. In other words, when we add machines a coalition receives the gains its members make in waiting costs, as the queue moves up, but must fully pay for the new machines. 

To properly express how this requeueing occurs, we build from the initial schedule $\sigma_0=(\varphi_0,s_0)$ a priority order $\pi$, which will allow us to determine, which agent moves up when new machines becomes available. Formally, for any $i,j\in N$
\[
\pi(i)<\pi(j)\Leftrightarrow \bigl\{s_{0}(i)< s_{0}(j) \text{ or } \{s_{0}(i)= s_{0}(j) \text{ and }\varphi_{0}(i)<\varphi_{0}(j)\}\bigr\}.
\]

In words, to rank agents we first look at the period in which they are served, and break ties by giving priority to agents served on machines identified with lower numbers.


Notice that in a public requeueing game a coalition $T$ has much less ability to choose an alternative schedule. Once it has chosen a new number of machines, agents requeue automatically using the ordering $\pi$. Coalition $T$ however can still reorder its members, under constraint. We assume that they can do so at two occasions, before and after adjusting the number of machines, under the same constraints (with or without swaps) as for private games. We define as $\hat{\Sigma}^s$ and $\hat{\Sigma}^{ns}$ the admissible schedules under these constraints.

In a public requeuing game (public game for short), given that the machines are public goods, we now suppose that the revenues from the sale of machines must be split equally among all agents in $N$, as the machines are public. Coalition $T$ thus receives a fraction $\frac{|T|}{n}$ of the proceeds. 

Let $\hat{V}_{s}(T,k)$ and $\hat{V}_{ns}(T,k)$ be the functions giving the value (possibly negative) that we obtain if we force coalition $T$ to use $k$ machines, in the public game with and without swaps, respectively. We then have that

\[
\hat{V}_{s}(T,k)=\left\{ 
\begin{array}{c}
\max_{\sigma \in \hat{\Sigma}_{T}^{s}(k)}\left( \sum_{i\in T}\left(
s_{0}(i)\right) -s(i))w_{i}-(k-m_{0})b\right) \text{ if }k\geq m_{0} \\ 
\max_{\sigma \in \hat{\Sigma}_{T}^{s}(k)}\left( \sum_{i\in T}\left(
s_{0}(i)\right) -s(i))w_{i}-\frac{\left\vert T\right\vert }{n}%
(k-m_{0})b\right) \text{ if }k<m_{0}%
\end{array}%
\right. 
\]

It is possible for $\hat{\Sigma}_{T}^{s}(k)$ to be empty if $k<m_0$, if $T$ does not have exclusive use of $m_0-k$ machines, in which case we simply let $\hat{V}_{s}(T,k)=0$. We define $\hat{V}_{ns}(T,k)$ in the same manner.

In the public setting, a new possibility emerges, in which a coalition $T$ can offer side-payments to agents in $N\setminus T$ to move further down the queue. We take an optimistic approach and suppose that the side-payments need to be just enough to cover the additional waiting costs of these agents. Let $\hat{V}_{sp}(T,k)$ be the function giving the value (possibly negative) that we obtain if we force coalition $T$ to use $k$ machines, in the public game with side payments.

To illustrate the differences between the approaches, consider four agents on two machines, with the queue being 1-3 on the first machine and 2-4 on the second machine. Consider coalition $\{1,2\}$. In the games without side payments, these agents cannot sell a machine because they are not the sole users of any of the two machines. Thus $\hat{\Sigma}^{NS}_{\{1,2\}}=\hat{\Sigma}^{S}_{\{1,2\}}=\emptyset$. With side payments however, the coalition can sell the second machine and use the queue 1-2-3-4 on the remaining machine, offering the proper compensation to agents 3 and 4 to move further down the queue. We obtain that $\hat{V}_{sp}(\{1,2\},1)=\frac{1}{2}b-w_2-w_3-2w_4$. Consider now coalition $\{1,2,3\}$. Without side-payments and without swaps, the coalition can sell the first machine and use the queue 2-4-1-3 on the remaining one, for a value of $\hat{V}_{ns}(\{1,2,3\},1)=\frac{3}{4}b-2w_1-2w_3$. With swaps but no side-payments, the coalition can now use the queue 1-4-2-3 for a value of $\hat{V}_{s}(\{1,2,3\},1)=\frac{3}{4}b-2w_2-2w_3$. Finally, with side-payments, we can use the queue 1-2-3-4 and offer $2w_4$ to agent 4 to get her to move to the end of the queue. We obtain $\hat{V}_{sp}(\{1,2,3\},1)=\frac{3}{4}b-w_2-w_3-2w_4$. 

We obtain the optimal cost savings for a coalition by maximizing over the number of machines: $\hat{V}_{sp}(T)\equiv \max_{k\in\{1,...,n\}}\hat{V}_{sp}(T,k)$, $\hat{V}_s(T)\equiv \max_{k\in\{1,...,n\}}\hat{V}_s(T,k)$ and $\hat{V}_{ns}(T)\equiv \max_{k\in\{1,...,n\}}\hat{V}_{ns}(T,k)$.

In the example, we have that $\hat{V}_{sp}\ge \hat{V}_{s} \ge \hat{V}_{ns}$, a result that is general, and offered without proofs, as it simply depends on the set of possibilities.

\begin{proposition}\label{relationship_public_games}
Let $(N,m_{0},\sigma_0,w,b)$ be a public requeueing problem. Then, $\hat{V}_{sp}\ge \hat{V}_{s} \ge \hat{V}_{ns}$.
\end{proposition}




A \emph{public requeueing problem} is a requeueing problem $(N,m_0,\sigma_0,w,b)$ with an endogenous number of machines where the machines are public goods. If all machines are public goods, we call a requeueing game with an endogenous number of machines a \emph{public requeueing game} denoted by $(N,\hat{V}_{sp})$, $(N,\hat{V}_{s})$ or $(N,\hat{V}_{ns})$, depending if we allow or not side payments and swaps.

It is difficult to offer general formulas for $\hat{V}_{s}(T)$ or $\hat{V}_{ns}(T)$, as the sets $\hat{\Sigma}^{ns}_T$ and $\hat{\Sigma}^{s}_T$ have a structure highly dependent on the initial ordering. However, $(N,\hat{V}_{sp})$ is much easier to express, as it is easy to show that if coalition $S$ sells machines and uses only $k$ machines, it is always optimal to move to the optimal queue, offering side-payments to non-members to achieve the result. Let $\hat{m}$ be the function assigning to each coalition the optimal number of machines to use when in the public game with side-payments.\footnote{There could be many, in which case we pick the lowest one.} We offer results on the structure of $\hat{m}$ when we start with an efficient initial ordering.

\begin{lemma}
\label{lemma_m}
Let $(N,m_{0},\sigma_0,w,b)$ be a public requeueing problem such that the ordering $\pi$ induced by $\sigma_0$ is the optimal queue $(1,2,...,n)$. Then, we have:

i) for all $S,T\subseteq N,$ $(\hat{m}(S)-m_{0})(\hat{m}(T)-m_{0})\geq 0.$

ii) if $S\subset T\subseteq N,$ then $\left\vert \hat{m}(S)-m_{0}\right\vert \leq
\left\vert \hat{m}(T)-m_{0}\right\vert $.
\end{lemma}

In words, part i) confirms that we cannot have some coalition buying machines while others sell machines. Either all coalitions buy machines (or stay put) or all coalitions sell machines (or stay put). Part ii) says that if $S$ is a subset of $T$, $T$ will make at least as many transactions as $S$: if $S$ buys some machines, $T$ will buy at least as many, and if $S$ sells some machines, $T$ will sell at least as many.

This structure allows us to  guarantee the non-emptiness of the core for public requeueing games with side-payments when the initial queue is optimal. 

\begin{theorem}
\label{thm:public_convex}
Let $(N,m_{0},\sigma_0,w,b)$ be a public requeueing problem such that the ordering $\pi$ induced by $\sigma_0$ is the optimal queue $(1,2,...,n)$, and let $(N,\hat{V}_{sp})$ be the associated public requeueing game with side payments. Then, $Core(\hat{V}_{sp})\neq\emptyset$.
\end{theorem}

\begin{proof}
 To ease on the notation in the proof, we use $V$ for $\hat{V}_{s}$. Let $\hat{N}_{1}$ be the set of agents that have their job processed by the end of period 1 in the initial situation with $m_0$ machines, i.e. $\hat{N}_{1}=\{1,...,m_0\}$. 
 Let $[i]=\{1,\ldots, i\}$ be the set of agents in the queue until agent $i$.

We will construct a core allocation $y$ to prove the non-emptiness of the core. Let $y$ be an allocation such that $y_{i}=\frac{V(\hat{N}_{1})}{|N_{1}|}$ if $i\in\hat{N}_{1}$ and $y_{i}=V([i])-V([i-1])$ otherwise. For a given coalition $S\subseteq N$, we will distinguish between two cases, namely, (i) when the coalition buys more machines and (ii) the coalition sells machines.

First, let us consider the case where the coalition $S$ buys machines. Notice that agents in $\hat{N}_{1}$ are not interested in buying new machines and bring no additional values when we do so. Thus, it sufficient to check the core constraints for $S\subseteq N\setminus \hat{N}_{1}$.

By Lemma \ref{lemma_m}, $\hat{m}([i-1])\leq \hat{m}([i])$ for all $i\in N$. For $i\in N\setminus \hat{N}_{1}$, we have that $$y_{i}=\left(\left\lceil\frac{i}{m_{0}}\right\rceil - \left\lceil\frac{i}{\hat{m}([i-1])}\right\rceil \right)w_{i}+\max_{k=m[i-1]}^{n}\left(\sum\limits_{l=1}^{i}\left(\left\lceil\frac{l}{\hat{m}([i-1])}\right\rceil - \left\lceil\frac{l}{k}\right\rceil \right)w_{l}-b(k-\hat{m}([i-1])\right).$$

Suppose that coalition $S$ uses $k$ machines, with $m_{0}\le k\le n$. Then, $V(S)=\sum\limits_{i\in S}\left(\left\lceil \frac{i}{m_{0}} \right\rceil-\left\lceil \frac{i}{k} \right\rceil\right)w_{i}-b(k-m_{0})$. 

Let $S^k=\{i\in S | \hat{m}([i-1])\ge k\}$, the set of agents $i\in S$ such that coalition $\{1,...,i-1\}$ uses at least $k$ machines. Then, for every such agent $i$, $y_{i}\ge \left(\left\lceil \frac{i}{m_{0}}\right\rceil-\left\lceil\frac{i}{\hat{m}([i-1])}\right\rceil\right)w_{i}\ge \left(\left\lceil \frac{i}{m_{0}}\right\rceil-\left\lceil\frac{i}{k}\right\rceil\right)w_{i} $, and we get, by summing up, that $$\sum_{i\in S^k} y_{i}\ge \sum_{i\in S^k} \left(\left\lceil \frac{i}{m_{0}}\right\rceil-\left\lceil\frac{i}{k}\right\rceil\right)w_{i}. $$

If $S=S^k$, then we are done. Suppose the contrary. Then, there exists an agent in $i\in S$ for which $\hat{m}([i-1])< k$. 

Take the agent $i\in S\setminus S^k$ with the largest index\footnote{By Lemma \ref{lemma_m}, if $i\in S^k$, then $i>j$ for all $j\in S\setminus S^k$. In words, the agents in $S^k$ are further down in the queue and have larger waiting costs than those in $S\setminus S^k$.} and let $k'=\hat{m}([i-1])<k$. Then,
\begin{align*}
    y_{i}&\ge \left(\left\lceil\frac{i}{m_{0}}\right\rceil - \left\lceil\frac{i}{k'}\right\rceil \right) w_{i}+   \sum\limits_{l=1}^{i}\left(\left\lceil\frac{l}{k'}\right\rceil - \left\lceil\frac{l}{k}\right\rceil \right) w_{l}-b(k-k')\\
    &\ge \left(\left\lceil\frac{i}{m_{0}}\right\rceil - \left\lceil\frac{i}{k'}\right\rceil \right) w_{i}+   \sum_{\substack{l\in S\\l\le i}}\left(\left\lceil\frac{l}{k'}\right\rceil - \left\lceil\frac{l}{k}\right\rceil \right) w_{l}-b(k-k').
\end{align*}
For any other $j\in S$ such that $\hat{m}([j-1])=k'$, we have that $y_{j}\ge \left(\left\lceil\frac{j}{m_{0}}\right\rceil-\left\lceil\frac{j}{k'}\right\rceil\right)w_{j}$. Together with the inequality for the agent with the largest index, we obtain that $$\sum_{i\in S^{k'}}y_{i}\ge \sum_{i\in S^{k'}}\left(\left\lceil\frac{i}{m_{0}}\right\rceil-\left\lceil\frac{i}{k}\right\rceil\right)w_{i}+\sum_{i\in S\setminus S^{k'}}\left(\left\lceil\frac{i}{k'}\right\rceil-\left\lceil\frac{i}{k}\right\rceil\right)w_{i}-b(k-k').$$
If $S=S^{k'}$, we are done. If not, repeat the procedure. Notice that, at each step, if we move from $k''$ to $k'''$, we end up with $$\sum_{i\in S^{k'''}}y_{i}\ge \sum_{i\in S^{k'''}}\left(\left\lceil\frac{i}{m_{0}}\right\rceil-\left\lceil\frac{i}{k}\right\rceil\right)w_{i}+\sum_{i\in S\setminus S^{k'''}}\left(\left\lceil\frac{i}{k'''}\right\rceil-\left\lceil\frac{i}{k}\right\rceil\right)w_{i}-b(k-k''').$$ In the end, we reach a point where $S^{k'''}=S$ for some $k'''\ge m_0$. In this case, the inequality simplifies to 
\begin{align*}
y_{S}&\ge \sum\limits_{i\in S} \left(\left\lceil\frac{i}{m_{0}}\right\rceil - \left\lceil\frac{i}{k}\right\rceil \right)w_{i}-b(k-k''')\\
&\ge V(S).
\end{align*}

We have shown that $y(S)\ge V(S)$ when the coalition $S$ uses at least $m_0$ machines.  Next, we will show that when the coalition $S$ sells machines, we can construct a core allocation $y$ making use of the symmetric function to the buying case. By Lemma $\ref{lemma_m}$, if $S$ sells machines, then no coalition buys machines. In addition $\hat{m}([i])\le \hat{m}([i-1])$ for all $i\in N$. 

For the case of selling machines, we define the allocation $y$ in the following way: for all agents whose operation is processed by the first period, that is, for all $i\in \hat{N}_{1}$, $y_{i}=\frac{V(\hat{N}_1)}{|\hat{N}_1|}=\frac{b}{n}\left( m_{0}-\hat{m}(%
\hat{N}_{1})\right) -\frac{1}{\left\vert \hat{N}_{1}\right\vert }\sum_{j\in
N}\left( \left\lceil \frac{j}{m_{0}}\right\rceil -\left\lceil \frac{j}{\hat{m}(%
\hat{N}_{1})}\right\rceil \right) w_{j}$, and $y_{i}=%
V(\left[i\right])-V(\left[i-1\right])=\frac{b}{n}\left( m_{0}-\hat{m}(\left[ i-1\right] )\right) +\max_{k=1}^{\hat{m}([i-1])} \left( \frac{ib}{n}(\hat{m}([i-1])-k)-\sum\limits_{j\in N}\left( \left\lceil \frac{j}{m\left( \left[ i-1\right]
\right) }\right\rceil -\left\lceil \frac{j}{k}\right\rceil \right) w_{j}\right) $ otherwise.

Suppose that coalition $S$ uses $k$ machines, with $1\le k \le m_0$.  The value obtained by coalition $S$ if it sells $k$ machine is $V(S)=\frac{|S|}{n}b(m_0-k)-\sum\limits_{i=1}^{n}\left(\left\lceil\frac{i}{m_{0}-k}\right\rceil-\left\lceil\frac{i}{m_{0}}\right\rceil\right)w_{i}$.

First, suppose that $\hat{m}(\hat{N}_1)\leq k$. Then, by Lemma \ref{lemma_m} 
there does not exists $i\in S\setminus \hat{N}_{1}$ such that $\hat{m}(\left[ i-1\right] )>k.$ Then, for all agents in $S$ whose jobs are not processed by the first period, that is $i\in S\backslash \hat{N}_{1},$ we have $y_{i}\geq \frac{b}{n}\left( m_{0}-k\right) $ while for the agents whose jobs are processed after the first period we have $\sum\limits_{i\in S\cap \hat{N}_{1}}y_{i}\geq
\left\vert S\cap \hat{N}_{1}\right\vert \frac{b}{n}\left( m_{0}-k\right) -%
\frac{\left\vert S\cap \hat{N}_{1}\right\vert }{\left\vert \hat{N}%
_{1}\right\vert }\sum\limits_{j\in N}\left( \left\lceil \frac{j}{m_{0}}\right\rceil
-\left\lceil \frac{j}{k}\right\rceil \right) w_{j}$.

Summing up the payoffs of all agents $i\in S\setminus \hat{N}_{1}$ and all agents $i\in S\cap \hat{N}_{1}$, we obtain
\begin{eqnarray*}
y_{S} &\geq &\left\vert S\right\vert \frac{b}{n}\left( m_{0}-k\right) -%
\frac{\left\vert S\cap \hat{N}_{1}\right\vert }{\left\vert \hat{N}%
_{1}\right\vert }\sum_{j\in N}\left( \left\lceil \frac{j}{m_{0}}\right\rceil
-\left\lceil \frac{j}{k}\right\rceil \right) w_{j} \\
&\geq &\left\vert S\right\vert \frac{b}{n}\left( m_{0}-k\right)
-\sum_{j\in N}\left( \left\lceil \frac{j}{m_{0}}\right\rceil -\left\lceil 
\frac{j}{k}\right\rceil \right) w_{j} \\
&=&V(S).
\end{eqnarray*}

Next, assume that $\hat{m}(\hat{N}_1)>k$.  We have that $\sum\limits_{i\in S\setminus S^{k+1}}y_{i}\geq \left\vert S\setminus S^{k+1}\right\vert \frac{b}{n}\left( m_{0}-k\right)$.

We distinguish two possible cases:\vspace{0.2cm}\\
\noindent \textbf{(i)} $S^{k+1}\subseteq\hat{N}_{1}.$\\

We have that  $\sum\limits_{i\in S\cap \hat{N}_{1}}y_{i}=\sum\limits_{i\in S^{k+1}} y_{i}\geq \left\vert S\cap \hat{N}%
_{1}\right\vert \frac{b}{n}\left( m_{0}-k
\right) -\frac{\left\vert S\cap \hat{N}_{1}\right\vert }{\left\vert \hat{N}%
_{1}\right\vert }\sum_{j\in N}\left( \left\lceil \frac{j}{m_{0}}\right\rceil
-\left\lceil \frac{j}{k}\right\rceil
\right) w_{j}.$

We thus obtain that 
\begin{eqnarray*}
y_{S} &=& \sum_{i\in S\setminus S^{k+1}} y_{i}+ \sum_{i\in S^{k+1}} y_{i} \\
&\geq &\left\vert S\right\vert \frac{b}{n}\left( m_{0}-k\right)
-\sum_{j\in N}\left( \left\lceil \frac{j}{m_{0}}\right\rceil -\left\lceil 
\frac{j}{k}\right\rceil \right) w_{j} \\
&=&V(S).
\end{eqnarray*}
\vspace{0.2cm}\\
\noindent \textbf{(ii)} $S^{k+1}\nsubseteq\hat{N}_{1}.$\\

Take now the member of $S^{k+1}$ with the largest index,
call him agent $i$ and let $k'=\hat{m}([i-1])>k$. We have that 
\begin{equation*}
y_{i}\geq \frac{b}{n}\left( m_{0}-k' )\right) +i\frac{b}{n}%
(k'-k)-\sum_{j\in N}\left( \left\lceil \frac{j}{%
k'}\right\rceil -\left\lceil \frac{j}{%
k }\right\rceil \right) w_{j}.
\end{equation*}

Notice that there are at most $i$ agents in $S^{k+1}$ before $i,$ and thus we have
covered, for all of them, the fraction of value $\frac{b}{n}(k'-k).$

We now have that 
\begin{align}
y_{S} & \geq  \sum_{i\in S\setminus S^{k'+1}}\left\vert
S\setminus S^{k'+1}\right\vert \frac{b}{n}\left(
m_{0}-k\right) -\sum_{j\in N}\left( \left\lceil \frac{j}{k'}\right\rceil -\left\lceil \frac{j}{k}%
\right\rceil \right) w_{j} \nonumber \\
&+ \sum_{i\in S^{k'+1}}\left(y_{i}+\frac{%
b}{n}(k'-k)\right). \label{c} \tag{c} \nonumber
\end{align}

Notice that if $S^{k'+1}=\emptyset,$ we are done, as $y_{S}\geq
V(S).$ Otherwise, repeat the process with $S^{k'+1}$. If $S^{k'+1}\subseteq\hat{N}_{1},$ repeat part (i)  to conclude that  $\sum_{i\in S^{k'+1}} y_{i}\geq \left\vert S\cap \hat{N}%
_{1}\right\vert \frac{b}{n}\left( m_{0}-k'
\right) -\frac{\left\vert S\cap \hat{N}_{1}\right\vert }{\left\vert \hat{N}%
_{1}\right\vert }\sum_{j\in N}\left( \left\lceil \frac{j}{m_{0}}\right\rceil
-\left\lceil \frac{j}{k'}\right\rceil
\right) w_{j}$, and combined with (\ref{c}), that $y(S)\ge V(S)$. If not, repeat part (ii). Notice that, at each step, if we move from k'' to k''', we end up with
\begin{align}
y_{S} & \geq  \sum_{i\in S\setminus S^{k'''+1}}\left\vert
S\setminus S^{k'''+1}\right\vert \frac{b}{n}\left(
m_{0}-k\right) -\sum_{j\in N}\left( \left\lceil \frac{j}{k'''}\right\rceil -\left\lceil \frac{j}{k}%
\right\rceil \right) w_{j} \nonumber \\
&+ \sum_{i\in S^{k'''+1}}\left(y_{i}+\frac{%
b}{n}(k'''-k)\right). \nonumber
\end{align}

The procedure wraps up in a finite number of steps, and thus, $y_{S}\geq V(S).$

\end{proof} 


Given Proposition \ref{relationship_public_games}, the following corollary is immediate.

\begin{corollary}
Let $(N,m_{0},\sigma_0,w,b)$ be a public requeueing problem such that the ordering $\pi$ induced by $\sigma_0$ is the optimal queue $(1,2,...,n)$, and let $(N,\hat{V}_{s})$ and $(N,\hat{V}_{ns})$ be respectively the associated public requeueing game with and without swaps. Then, $Core(\hat{V}_{s})\subseteq Core(\hat{V}_{ns})\neq\emptyset$.
\end{corollary}

An important consequence of Theorem \ref{thm:public_convex} is that whenever each agents own a machine at the initial schedule, then the core is always non-empty.

\begin{corollary}
\label{cor:public_each_own_m}
Given a public requeueing problem $(N,m_0,\sigma_0,w,b)$ such that $m_0=|N|$, the associated public requeueing game without swaps $(N,\hat{V}_{ns})$, with swaps $(N,\hat{V}_{s})$, and with side-payments $(N,\hat{V}_{sp})$ have a non-empty core.
\end{corollary}

The assumption that agents are ranked in an optimal way in the original schedule is crucial in our proof, allowing us to obtain the same structure for the optimal number of machines as when queueing without an initial allocation. For instance, if agent 1 has a particularly large waiting cost and is initially ranked last, when he is by himself he might prefer to buy many machines to be served earlier, while with other agents it might not be necessary, as switching spots with other members of the coalition might allow him to be served early without buying as many machines. 

The next example shows that if we relax the assumption that the initial queue is optimal, the core of a public requeueing game can be empty.

\begin{example}
\label{ex: public_no_opt_emptycore}
Consider $(N,m_0,\sigma_0,w,b)$ with $N=\{1,2,3,4\}$, $m_0=1$, and the ordering induced by $\sigma_0$ being $\pi=(4,3,2,1)$. The waiting costs per unit for agents are $w=(w_{i})_{i\in N}=(13,7,6,1)$, and the cost of a machine is $b=15$.

Notice first that the initial queue $(4,3,2,1)$ is not optimal:

\begin{center}
\begin{tabular}{c|c|c|c|c|}
\cline{2-5}
$m_{1}$ & 4 & 3 & 2 & 1 \\ \cline{2-5}
\end{tabular}.%
\end{center}
Suppose that we allow agents to jump over the agents not belonging to the coalition. 
Let $(N,\hat{V}_{s})$ be the associated public requeueing game. Let us illustrate how to calculate the value of a coalition by doing it for $\{1,4\}$. If agent 1 and agent 4 do not buy any machine, then they switch their positions:
\begin{center}
\begin{tabular}{c|c|c|c|c|}
\cline{2-5}
$m_{1}$ & 1 & 3 & 2 & 4 \\ \cline{2-5}
\end{tabular}, %
\end{center}
and hence the total savings for $\{1,4\}$ is $(13-1)\times3=36$. 

If they buy a new machine, since it is a public requeueing game, the queue moves in a way such that agents 3, 4 are served first and agents 1, 2 are served second. Then, agents 1 and 4 change positions:
\begin{center}
\begin{tabular}{c|c|c|c|c|}
\cline{2-5}
$m_{1}$ & 3 & 2 & \,\, & \,\,  \\ \cline{2-5}
$m_{2}$ & 1 & 4 & \,\, & \,\,  \\ \cline{2-5}
\end{tabular},
\end{center}
and hence the total savings for $\{1,4\}$ is $13\times3-15-1=23$.

If they buy two new machines, the queue moves in a way such that agents 2,3,4 are served in the first position of a machine while agent 1 is served in the second position at a machine. Then, agents 1 and 4 switch positions:
\begin{center}
\begin{tabular}{c|c|c|c|c|}
\cline{2-5}
$m_{1}$ & 1 & 4 & \,\, & \,\, \\ \cline{2-5}
$m_{2}$ & 3 &   & &  \\ \cline{2-5}
$m_{3}$ & 2 &  & &  \\ \cline{2-5}
\end{tabular},
\end{center}
and hence the total savings for $\{1,4\}$ is $13\times 3-15\times 2-1=8$. Finally, if they buy three machines, all agents' operations are processed at the first position on each machine:
\begin{center}
\begin{tabular}{c|c|c|c|c|}
\cline{2-5}
$m_{1}$ & 1 & \,\, & \,\, & \,\, \\ \cline{2-5}
$m_{2}$ & 2 & & &  \\ \cline{2-5}
$m_{3}$ & 3 & & &  \\ \cline{2-5}
$m_{4}$ & 4 & & &  \\ \cline{2-5}
\end{tabular},
\end{center}
and hence the total savings is $13\times 3-15\times 3 = -6$. Then, the maximum of possible total cost savings for $\{1,4\}$ is achieved when they do not buy any new machine, $\hat{V}_{s}(\{1,4\})=\max\{36,23,8,-6\}=36$.

Consider now coalitions $\{1,4\}$, $\{2,4\}$, $\{3,4\}$, $\{1,2,3\}$, and $\{1,2,3,4\}$. Following our illustration, one can check that $\hat{V}_{s}(\{1,4\})=36$, $\hat{V}_{s}(\{2,4\})=12$, $\hat{V}_{s}(\{3,4\})=5$, $\hat{V}_{s}(\{1,2,3\})=31$, and $\hat{V}_{s}(N)=37$. Now, suppose that there exists a core allocation $(y_1,y_2,y_3,y_4)\in Core(\hat{V}_{s})$. Then, it would satisfy the core constraints $y_{1}+y_{4}\ge 36$, $y_{2}+y_{4}\ge 12$, $y_{3}+y_{4}\ge 5$, $y_{1}+y_{2}+y_{3}\ge 31$. Nevertheless, such a payoff vector $(y_1,y_2,y_3,y_4)$ is not in the core since the core constraints $(y_{1}+y_{4})+(y_{2}+y_{4})+(y_{3}+y_{4})+2(y_{1}+y_{2}+y_{3})=3(y_{1}+y_{2}+y_{3}+y_{4})\ge 115$ is not compatible with $3\hat{V}_{s}(N)=111$, and hence $Core(\hat{V}_{s})=\emptyset$.
\end{example}

\section{Concluding remarks}
\label{sec:conc}

This paper studies queueing problems from a game theoretical point of view. The novelty of this paper is that the number of machines is endogenous. For a given problem, agents are allowed to (de)activate as many machines they want, at a cost. We have distinguished two types of queueing problems: without and with an initial queue. For the first case, we have provided both a lower and an upper bound on the cost of machine to guarantee the non-emptiness of the core. Moreover, in some instances we have provided a full characterization of the core by means of concavity. For the second case, although we have shown that the core may be empty, we have guaranteed balancedness when all machines are accessible to all agents and the initial ordering correctly ranks agents in decreasing order of their waiting costs.  

The proof of that last result is constructive: the allocation used to show the result is, interestingly, a mix between an average-value and a marginal value allocation: agents served in the first period equally share the value they create together, while other agents are allocated their marginal contribution, precisely when they join in the optimal ordering, i.e. from largest to smallest waiting costs. In particular, if we start with a single machine, the allocation becomes the marginal contribution allocation. \footnote{\cite{vvh03} introduce a class of balanced games in a more general setting of sequencing problems with one machine, and show that the marginal contribution allocation is indeed stable.} It is important to note, however, that the resulting game with side-payments is not necessarily convex, and we cannot use any marginal contribution allocation, or even any of them if we start with more than one machine.

Compared to the earlier literature, our main innovations are (i) the existence of an endogenous number of machines at a given queueing problem, (ii) the cost associated with a machine to (de)activate it, (iii) the distinction between private and public queueing problems with an initial queue. 

An interesting direction for future research is to characterize axiomatically an allocation rule that always selects a stable allocation for balanced requeueing games. Furthermore, although we have a counterexample when swaps are allowed for public requeueing games with the non-optimal initial queue, it is still an open question whether it is also the case when swaps are not allowed.  


\appendix
\section{}
\label{sec:appendix}
We consign to this Appendix proofs of lemmata \ref{alpha}, \ref{subset}, \ref{lemma_m}.\vspace{0.2cm}

\noindent \emph{Proof of Lemma \ref{alpha}}. Fix $N$ and $w$. 
The total cost when $k$ machines are used would be cheaper than when $k-1$ machines are used if
\[
bk+\sum_{i\in N}\left( \left\lceil \frac{i}{k}\right\rceil \right) w_{i}\leq
b(k-1)+\sum_{i\in N}\left( \left\lceil \frac{i}{k-1}\right\rceil \right)
w_{i}
\]%
which simplifies to 
\begin{align}
b &\leq \sum_{i\in N}\left( \left\lceil \frac{i}{k-1}\right\rceil
-\left\lceil \frac{i}{k}\right\rceil \right) w_{i} \nonumber \\
&=w_{k}+\sum_{i=k+1}^{n}\left( \left\lceil \frac{i}{k-1}\right\rceil
-\left\lceil \frac{i}{k}\right\rceil \right) w_{i}. \tag{a} \label{def:r}
\end{align}%

The inequality (\ref{def:r}) provides an upper-bound on the cost of a machine such that we prefer to use $k$ machines to $k-1$ machines. Let us denote this number obtained in (\ref{def:r}) by $r^w(k)$. This defines a function $r^w:\{2,...,n\}\to \mathbb{R}_+$.


We next show that this function is non-increasing. We show that $r^w(k)\leq r^w(k-1),$ that is,%
\[
w_{k}+\sum_{i=k+1}^{n}\left( \left\lceil \frac{i}{k-1}\right\rceil
-\left\lceil \frac{i}{k}\right\rceil \right) w_{i}\leq
w_{k-1}+\sum_{i=k}^{n}\left( \left\lceil \frac{i}{k-2}\right\rceil
-\left\lceil \frac{i}{k-1}\right\rceil \right) w_{i}. \label{ineq}\tag{b}
\]

By assumption, $w_{k}\leq w_{k-1}.$ We will show that \[ \sum_{i=k+1}^{n}\left( \left\lceil \frac{i}{k-1}\right\rceil
-\left\lceil \frac{i}{k}\right\rceil \right) w_{i}\leq
\sum_{i=k}^{n}\left( \left\lceil \frac{i}{k-2}\right\rceil
-\left\lceil \frac{i}{k-1}\right\rceil \right) w_{i}, \] which together with $w_{k}\leq w_{k-1}$ show that the inequality holds. To do so, we compare the right-hand side and the left-hand side summands of the same order in the inequality (\ref{ineq}). We see that
\begin{eqnarray*}
\left( \left\lceil \frac{k+1}{k-1}\right\rceil
-\left\lceil \frac{k+1}{k}\right\rceil \right) w_{k+1}&\leq&
\left( \left\lceil \frac{k}{k-2}\right\rceil
-\left\lceil \frac{k}{k-1}\right\rceil \right) w_{k}\\
\left( \left\lceil \frac{k+2}{k-1}\right\rceil
-\left\lceil \frac{k+2}{k}\right\rceil \right) w_{k+2}&\leq&
\left( \left\lceil \frac{k+1}{k-2}\right\rceil
-\left\lceil \frac{k+1}{k-1}\right\rceil \right) w_{k+1}\\
& \vdots & \\
\left( \left\lceil \frac{n-1}{k-1}\right\rceil
-\left\lceil \frac{n-1}{k}\right\rceil \right) w_{n-1}&\leq&
\left( \left\lceil \frac{n-2}{k-2}\right\rceil
-\left\lceil \frac{n-2}{k-1}\right\rceil \right) w_{n-2}\\
\left( \left\lceil \frac{n}{k-1}\right\rceil
-\left\lceil \frac{n}{k}\right\rceil \right) w_{n}&\leq&
\left( \left\lceil \frac{n-1}{k-2}\right\rceil
-\left\lceil \frac{n-1}{k-1}\right\rceil \right) w_{n-1}
\end{eqnarray*}
\begin{eqnarray*}
0&\leq&
\left( \left\lceil \frac{k}{k-2}\right\rceil
-\left\lceil \frac{k}{k-1}\right\rceil \right) w_{n},
\end{eqnarray*}
and we see that removing a machine is costlier in terms of waiting costs if there are less machines in the initial problem. Applying the result recursively, starting with $r^w(n)$, we obtain that $r^w$ is non-increasing.

It remains to show that we can define $m$ using $r^w$. Let $C(N,k)$ be the cost for coalition $N$ if it uses $k$ machines. Suppose that $b\geq r^w(2)$. Then, since $r^w$ is non-increasing $b\geq r^w(k)$ for all $k\in \{2,...,n\}.$ This implies that $C(N,k)\leq C(N,k+1)$ for all $k=1,...,n-1.$ By transitivity, $C(N,1)\leq C(N,k)$ for all $k\in \{2,...,n\}$ and thus $m(N)=1$.

Suppose next that $r^w(k)>b\geq r^w(k+1)$ for some $1<k<n$. By the same argument as above, $b\geq r^w(k+1)$ implies that $C(N,k)\leq C(N,l)$ for all $l\in \{k+1,...,n\}$. Since $r^w$ is non-increasing, $r^w(k)>b$ implies that $r^w(l)>b$ for all $l=2,...,k$. This implies that  $C(N,l)< C(N,l-1)$ for all $l=2,...,k$. By transitivity, $C(N,k)< C(N,l)$ for all $l\in \{1,...,k-1\}$. Combining with the previous result, we obtain $m(N)=k$.

Finally, suppose that $r^w(n)>b$. By the same argument as above, we have that $C(N,n)< C(N,l)$ for all $l\in \{1,...,n-1\}$ and we obtain $m(N)=n$. \hfill \qed
\vspace{0.2cm}\\
\noindent \emph{Proof of Lemma \ref{subset}}. Let $r_{S}^{w}(k)$ be the equivalent of $r^{w}(k)$ for coalition $S.$

i) If $m(T)\geq \left\vert S\right\vert ,$ the result is immediate. Thus, suppose that $m(T)<\left\vert S\right\vert .$

We show that for any $S\subset T\subseteq N$ and $k=2,...,\left\vert
S\right\vert ,$ we have that $r_{S}^{w}(k)\leq r_{T}^{w}(k).$ That is,
\begin{eqnarray*}
r_{S}^{w}(k) &=&w_{k}^{S}+\sum_{l=k+1}^{\left\vert S\right\vert }\left(
\left\lceil \frac{l}{k-1}\right\rceil -\left\lceil \frac{l}{k}\right\rceil
\right) w_{l}^{S} \\
&\leq &w_{k}^{T}+\sum_{l=k+1}^{\left\vert S\right\vert }\left(
\left\lceil \frac{l}{k-1}\right\rceil -\left\lceil \frac{l}{k}\right\rceil
\right) w_{l}^{T} \\
&\leq &w_{k}^{T}+\sum_{l=k+1}^{\left\vert T\right\vert }\left(
\left\lceil \frac{l}{k-1}\right\rceil -\left\lceil \frac{l}{k}\right\rceil
\right) w_{l}^{T} \\
&=&r_{T}^{w}(k),
\end{eqnarray*}%
where the first inequality comes from the fact that $w_{k}^{S}\leq w_{k}^{T}$ for all $k.$

Then, if $b\geq r_{T}^{w}(2),$ $b\geq r_{S}^{w}(2)$ and $m(S)=m(T)=1.$ Otherwise, $m(S)$ is the highest integer such that $b<r_{S}^{w}(m(S)).$ But since $r_{S}^{w}(m(S))\leq r_{T}^{w}(m(S)),$ we have $b<r_{S}^{w}(m(T))$, and thus $m(S)\leq m(T),$ as desired.

ii) The proof is identical to part i), replacing $S$ by $S\cup \{i\}$ and $T$ by $S\cup \{j\}$. \hfill \qed
\vspace{0.2cm}\\
\noindent \emph{Proof of Lemma \ref{lemma_m}}. i) First, it is immediate that if a coalition prefers to buy $k>1$ machines
than use $m_{0}$ machines, it also prefers to buy one machine to using $m_{0}
$ machines. In the same way, if a coalition prefers to sell $k>1$ machines
to using $m_{0}$ machines, it also prefers to sell one machine to using $%
m_{0}$ machines. Thus, we only need to show that there cannot be $%
S,T\subseteq N$ such that $S$ prefers to buy a machine to using $m_{0}$
machines and $T$ prefers to sell a machine to using $m_{0}$ machines.

Suppose first that $S$ prefers to buy a machine to using $m_{0}$ machines.
Thus, $\sum_{i\in S}\left( \left\lceil \frac{i}{m_{0}}\right\rceil
-\left\lceil \frac{i}{m_{0}+1}\right\rceil \right) w_{i}-b>0.$ But, we have
that 
\begin{eqnarray*}
\sum_{i\in S}\left( \left\lceil \frac{i}{m_{0}}\right\rceil -\left\lceil 
\frac{i}{m_{0}+1}\right\rceil \right) w_{i}-b &\leq &\sum_{i\in N}\left(
\left\lceil \frac{i}{m_{0}}\right\rceil -\left\lceil \frac{i}{m_{0}+1}%
\right\rceil \right) w_{i}-b \\
&\leq &\sum_{i\in N}\left( \left\lceil \frac{i}{m_{0}-1}\right\rceil
-\left\lceil \frac{i}{m_{0}}\right\rceil \right) w_{i}-b \\
&\leq &\sum_{i\in N}\left( \left\lceil \frac{i}{m_{0}-1}\right\rceil
-\left\lceil \frac{i}{m_{0}}\right\rceil \right) w_{i}-\frac{\left\vert
T\right\vert }{n}b
\end{eqnarray*}%
and thus $\sum_{i\in N}\left( \left\lceil \frac{i}{m_{0}-1}\right\rceil
-\left\lceil \frac{i}{m_{0}}\right\rceil \right) w_{i}-\frac{\left\vert
T\right\vert }{n}b>0,$ which can be rewritten as $\frac{\left\vert
T\right\vert }{n}b-\sum_{i\in N}\left( \left\lceil \frac{i}{m_{0}-1}%
\right\rceil -\left\lceil \frac{i}{m_{0}}\right\rceil \right) w_{i}<0$ which
indicates that $T$ does not prefer to sell 1 machine to using $m_{0}$
machines.

Suppose next that $S$ prefers to sell a machine to using $m_{0}$ machines.
Thus, $\frac{\left\vert S\right\vert }{n}b-\sum_{i\in N}\left( \left\lceil 
\frac{i}{m_{0}-1}\right\rceil -\left\lceil \frac{i}{m_{0}}\right\rceil
\right) w_{i}>0.$ But, we have that 
\begin{eqnarray*}
\frac{\left\vert S\right\vert }{n}b-\sum_{i\in N}\left( \left\lceil \frac{i}{%
m_{0}-1}\right\rceil -\left\lceil \frac{i}{m_{0}}\right\rceil \right) w_{i}
&\leq &b-\sum_{i\in N}\left( \left\lceil \frac{i}{m_{0}-1}\right\rceil
-\left\lceil \frac{i}{m_{0}}\right\rceil \right) w_{i} \\
&\leq &b-\sum_{i\in N}\left( \left\lceil \frac{i}{m_{0}}\right\rceil
-\left\lceil \frac{i}{m_{0}+1}\right\rceil \right) w_{i} \\
&\leq &b-\sum_{i\in T}\left( \left\lceil \frac{i}{m_{0}}\right\rceil
-\left\lceil \frac{i}{m_{0}+1}\right\rceil \right) w_{i}
\end{eqnarray*}%
and thus $b-\sum_{i\in T}\left( \left\lceil \frac{i}{m_{0}}\right\rceil
-\left\lceil \frac{i}{m_{0}+1}\right\rceil \right) w_{i}>0,$ which can be
rewritten as  $\sum_{i\in T}\left( \left\lceil \frac{i}{m_{0}}\right\rceil
-\left\lceil \frac{i}{m_{0}+1}\right\rceil \right) w_{i}-b<0,$ which
indicates that $T$ does not prefer to buy 1 machine to using $m_{0}$
machines.

ii) Suppose that $S$ buys machines. Then, by part i), so does $T.$ We have
that $\sum_{i\in S}\left( \left\lceil \frac{i}{m_{0}}\right\rceil
-\left\lceil \frac{i}{\hat{m}(S)}\right\rceil \right) w_{i}-b\left(
\hat{m}(S)-m_{0}\right) \geq \sum_{i\in S}\left( \left\lceil \frac{i}{m_{0}}%
\right\rceil -\left\lceil \frac{i}{k}\right\rceil \right) w_{i}-b\left(
k-m_{0}\right) $ for all $k=m_{0},...,\hat{m}(S).$ Add $\sum_{i\in T\backslash
S}\left( \left\lceil \frac{i}{m_{0}}\right\rceil -\left\lceil \frac{i}{\hat{m}(S)}%
\right\rceil \right) w_{i}$ on both sides to obtain 
\[
\sum_{i\in T}\left( \left\lceil \frac{i}{m_{0}}\right\rceil -\left\lceil 
\frac{i}{\hat{m}(S)}\right\rceil \right) w_{i}-b\left( \hat{m}(S)-m_{0}\right) \geq
\sum_{i\in T}\left( \left\lceil \frac{i}{m_{0}}\right\rceil -\left\lceil 
\frac{i}{k}\right\rceil \right) w_{i}-b\left( k-m_{0}\right) 
\]%
for all  $k=m_{0},...,\hat{m}(S),$ and thus $T$ buys at least as many machines as $%
S$.

Suppose next that $S$ sells machines. Then, by part i), so does $T.$ We have
that $\frac{\left\vert S\right\vert }{n}b(m_{0}-\hat{m}(S))-\sum_{i\in N}\left(
\left\lceil \frac{i}{\hat{m}(S)}\right\rceil -\left\lceil \frac{i}{m_{0}}%
\right\rceil \right) w_{i}\geq \frac{\left\vert S\right\vert }{n}%
b(m_{0}-k)-\sum_{i\in N}\left( \left\lceil \frac{i}{k}\right\rceil
-\left\lceil \frac{i}{m_{0}}\right\rceil \right) w_{i}$ for all $%
k=\hat{m}(S),...,m_{0}.$ We then have that 
\[
\frac{\left\vert T\right\vert }{n}b(m_{0}-\hat{m}(S))-\sum_{i\in N}\left(
\left\lceil \frac{i}{\hat{m}(S)}\right\rceil -\left\lceil \frac{i}{m_{0}}%
\right\rceil \right) w_{i}\geq \frac{\left\vert T\right\vert }{n}%
b(m_{0}-k)-\sum_{i\in N}\left( \left\lceil \frac{i}{k}\right\rceil
-\left\lceil \frac{i}{m_{0}}\right\rceil \right) w_{i}
\]%
for all $k=\hat{m}(S),...,m_{0},$ and thus $T$ sells at least as many machines as $S.$ \hfill \qed

\bibliographystyle{te}
\bibliography{bib_AT}

\end{document}